\documentclass{article}

\usepackage{arxiv}

\usepackage[utf8]{inputenc} % allow utf-8 input
\usepackage[T1]{fontenc}    % use 8-bit T1 fonts
\usepackage{hyperref}       % hyperlinks
\usepackage{url}            % simple URL typesetting
\usepackage{booktabs}       % professional-quality tables
\usepackage{amsfonts}       % blackboard math symbols
\usepackage{nicefrac}       % compact symbols for 1/2, etc.
\usepackage{microtype}      % microtypography
\usepackage{lipsum}
\usepackage{graphicx}
\graphicspath{ {./images/} }
\usepackage{amssymb}
\usepackage{amsmath}
\usepackage{lineno,hyperref}
\usepackage{latexsym}
\usepackage{amsthm}
\usepackage{amssymb}
\usepackage{amsmath,float,cancel}
\usepackage{ulem}
\usepackage{tikz}
\usepackage{tkz-graph}
\usetikzlibrary{decorations,arrows,shapes}
\usepackage{enumerate}
\usepackage{comment}

\usepackage{amssymb,amsthm, amsfonts}
\usepackage{amsmath}
\usepackage{}
\usepackage{graphicx}
\usepackage{graphics}
\usepackage{tikz}
\usetikzlibrary{calc,positioning,fit, arrows,shapes}
\usepackage{float}
\usepackage{comment}
\usepackage{arydshln}
\usepackage{float}
\usepackage{subcaption}

\newtheorem{theorem}{Theorem}
\newtheorem{proposition}{Proposition}
\newtheorem{lemma}{Lemma}
\newtheorem{corollary}{Corollary}
\newtheorem{definition}{Definition}
\newtheorem{remark}{Remark}
\newtheorem{example}{Example}
\newtheorem{problem}{Problem}

\newcommand{\real}{{\rm I\!R}}

\title{New results on \texorpdfstring{$B_{\alpha}$}{Balpha}-eigenvalues of a graph}

\author{
  \textbf{Germain Pastén}\\
  Universidad de Antofagasta, Chile\\
  \texttt{germain.pasten@uantof.cl}
  \and
  \textbf{Carla S. Oliveira}\\
  ENCE/IBGE, Rio de Janeiro, Brazil\\
  \texttt{carla.oliveira@ibge.gov.br}\\
  \and
  \textbf{Jo\~ao Domingos G. da Silva Junior}\\
  Colégio Pedro II, Rio de Janeiro, Brazil\\
  \texttt{joao.dgomes@gmail.com}
  \and
  \textbf{Claudia M. Justel}\\
  IME, Rio de Janeiro, Brazil\\
  \texttt{cjustel@ime.eb.br}
}

\begin{document}
\maketitle
\begin{abstract}
Let $G$ be a graph with adjacency matrix $A(G)$ and Laplacian matrix $L(G)$. In 2024, Samanta \textit{et} \textit{al.}  defined the convex linear combination of $A(G)$ and $L(G)$ as $B_\alpha(G) = \alpha A(G) + (1-\alpha)L(G)$, for $\alpha \in [0,1]$. This paper presents some results on the eigenvalues of $B_{\alpha}(G)$ and their multiplicity when some sets of vertices satisfy certain conditions. Moreover,  the positive semidefiniteness problem of $B_{\alpha}(G)$ is studied.
\end{abstract}


\section{Introduction}

\vspace{-3mm}

Let $G=(V(G), E(G))$ be a simple undirected graph on $n$ vertices with vertex set $V(G)$ and edge set $E(G)$. If $\{u,v\} \in E(G)$, we say that $u$ is adjacent to $v$; otherwise, they are nonadjacent. A vertex $v$ is a neighbor of $u$ if it is adjacent to $u$. For each vertex $v \in V(G)$, the open neighborhood of $v$ in $G$, denoted by $N_{G}(v)$, is the set of all neighbors of $v$, while the closed neighborhood of $v$ is denoted and defined as $N_{G}[v]=N_{G}(v) \cup \{v\}$. The degree of $v$, denoted by $d_G(v)$, is the number of neighbors of $v$; that is, $|N_G(v)|=d_G(v)$. If the context is clear, $d_G(v)=d(v)$. A vertex $v \in V(G)$ is called a pendant vertex if $d_G(v)=1$. An independent set is a subset of vertices where there are no two vertices in the set that are adjacent. 
%{\color{red}{The chromatic number of a graph $G$, denoted by $\chi(G)$, is the smallest number of colors for $V(G)$ so that adjacent vertices are colored differently.}} 
The induced subgraph by $S \subseteq V(G),$ $G[S],$ is the graph whose vertex set is $S$  and whose edge set consists of all the edges in $E(G)$ that have both endpoints in $S.$ A clique $W \subseteq V(G)$ is a subset such that every two distinct vertices are adjacent; that is, the induced subgraph $G[W]$ is a complete graph. Throughout this work, the graphs will be considered to be of order greater than or equal to 2.

Let $D(G)$ be the diagonal matrix whose $(i,i)$-entry is the degree of the $i$-th vertex of $G$ and let $A\left( G\right)$ be the adjacency matrix of $G$. The Laplacian matrix and signless Laplacian matrix of $G$ are defined as $L(G)=D(G)-A(G)$ and $Q(G) =D(G) + A(G)$, respectively. In 2024, Samanta  \textit{et} \textit{al.} \cite{Samanta2024} defined the convex linear combination of $A(G)$ and $L(G)$ as 
\begin{equation*}
B_{\alpha}(G)= \alpha A(G) +(1-\alpha)L(G). 
\end{equation*}
for any real $\alpha \in [0,1]$. From \cite{Samanta2024}, $B_\alpha(G)$ can be rewritten as  
\begin{equation*}
B_\alpha(G)=\alpha D(G)+(1-2\alpha)L(G)
\end{equation*}
or
\begin{equation}\label{identity2}
B_\alpha(G)=(1-\alpha) D(G)+(2\alpha-1)A(G).
\end{equation}
It is easy to check that $D(G) = 2B_{\frac{1}{2}}(G),$ $ A(G) = B_{1}(G)$, $L(G) = B_{0}(G)$, and $Q(G) = 3 B_{\frac{2}{3}}(G)$. 

The definition of the $B_{\alpha}$-matrices was motivated by the convex linear combination $A_{\alpha}(G)= \alpha D(G) + (1 - \alpha) A(G),$ for $\alpha \in [0,1],$ introduced by Nikiforov \cite{Niki2017} in 2017. In that paper, the author presented a novel version of the spectral Turán theorem, until then unknown for $\alpha > \frac{1}{2}$. We highlight the following differences between the $A_{\alpha}$-matrices and the $B_{\alpha}$-matrices: 
\begin{itemize}
\item[(i)] The $A_{\alpha}$-matrices merge the theories of the adjacency matrix and signless Laplacian matrix of graphs, but does not directly include the theory of the Laplacian matrix. Indeed, for $0\leq \alpha < \beta \leq 1$, the Laplacian matrix appears in the identity 
\begin{equation*}
A_{\alpha}(G)-A_{\beta}(G)=(\alpha-\beta)L(G),
\end{equation*}
but not for any $\alpha$ in the convex combination. In contrast, the theory of $B_{\alpha}$-matrices merges the theories of these three matrices.

\item[(ii)] Unlike the $A_{\alpha}$-matrices, the $B_{\alpha}$-matrices are not always nonnegative, which poses challenges for their analysis.

\item[(iii)] Nikiforov \cite{Niki2017} proved the monotonicity of the eigenvalues of $A_{\alpha}(G)$ in $\alpha$, which does not occur for the eigenvalues of $B_{\alpha}(G)$, as can be seen in Figure~\ref{monotonicity}. 

\item[(iv)] Unlike the $A_{\alpha}$-matrices, which have been extensively studied since Nikiforov's seminal paper in 2017, the more recently introduced $B_{\alpha}$-matrices remain relatively unexplored.
\end{itemize}

In light of these differences and challenges, this work presents new results concerning the theory of $B_{\alpha}$-matrices.

\begin{figure}[H]
    \centering
    \begin{subfigure}[t]{0.4\textwidth}
        \centering
        \includegraphics[width=0.6\linewidth]{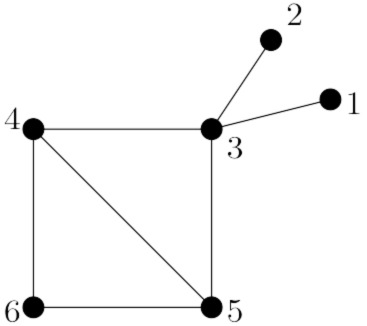}
        \caption{Graph $G$}
        \label{fig:graph}
    \end{subfigure}
\end{figure}
\begin{figure}[H]\ContinuedFloat
    \centering
    \begin{subfigure}[t]{0.85\textwidth}
        \centering
        \includegraphics[width=\linewidth]{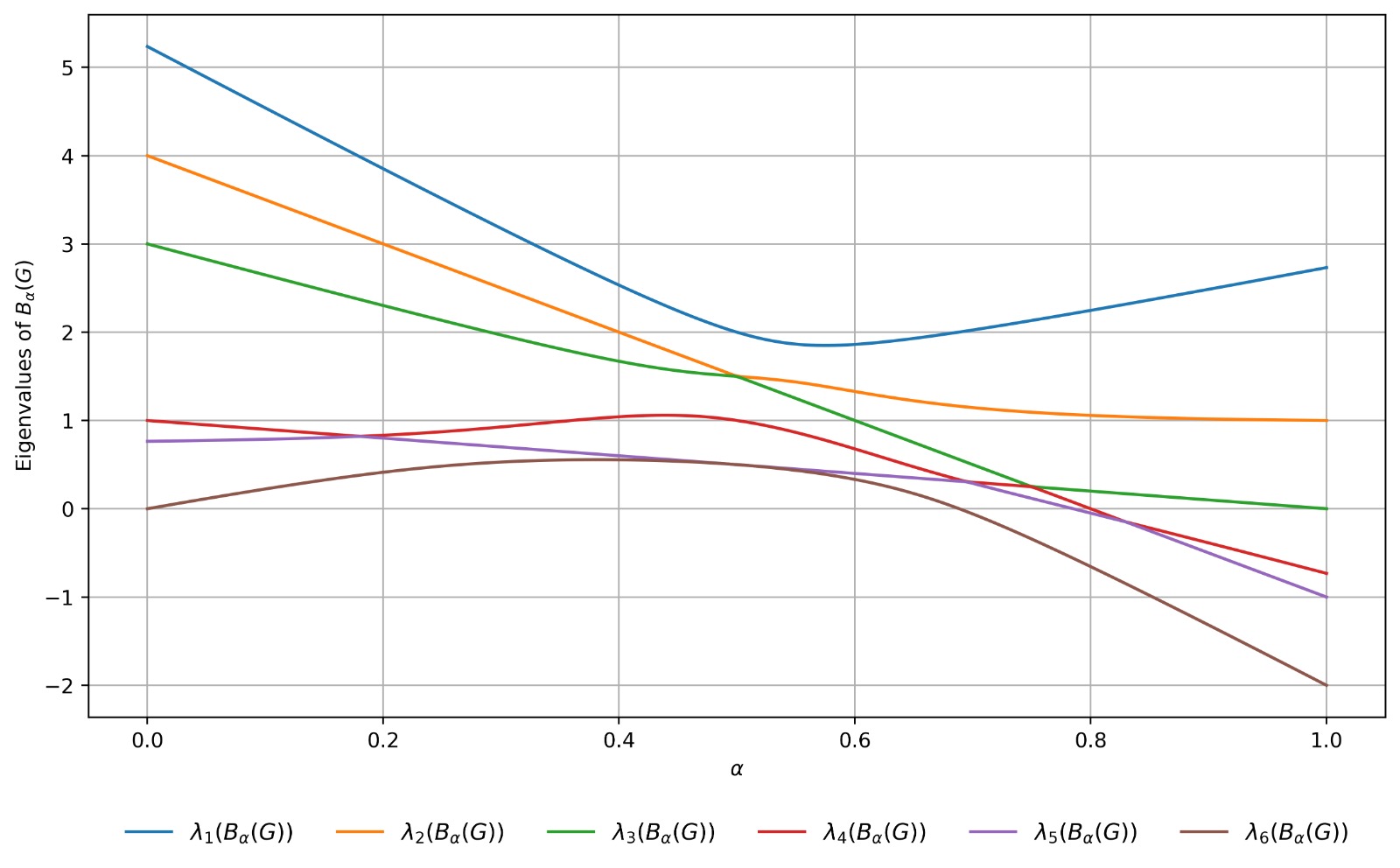}
        \caption{Behavior of the eigenvalues of $B_{\alpha}(G)$ as a function of $\alpha$.}
        \label{fig:graphics}
    \end{subfigure}
    \caption{An example of the nonmonotonicity of the eigenvalues of $B_{\alpha}(G)$ for $0 \leq \alpha \leq 1$.}
    \label{monotonicity}
\end{figure}

The paper is organized as follows: Section~\ref{prel} presents some preliminaries on Matrix Theory. Section~\ref{prop_B_alpha} introduces some properties of the $B_{\alpha}$-matrices of graphs. In Section~\ref{multiplicity}, we present some results involving the $B_{\alpha}$-eigenvalues of $G$ with their multiplicities when some sets of vertices satisfy certain conditions. Finally, in Section~\ref{semidef}, the positive semidefiniteness problem of $B_{\alpha}(G)$ is presented, and some results are exhibited.

%\vspace{-2mm}

\section{Preliminaries}\label{prel}

Let $M$ be a square matrix of order $n \times n$. The determinant of $M$ is denoted by $\left\vert M\right\vert$. Then, the $M$-characteristic polynomial is defined by 
\begin{equation*}
    P_M(\lambda) =  \vert \lambda I_n - M \vert
\end{equation*}
and the roots of $P_M(\lambda)$ are called the $M$-eigenvalues. The eigenvalues of an $M$-matrix associated with a graph $G$ will be referred to as the $M$-eigenvalues of $G$. If $M$ is real symmetric, its eigenvalues $\lambda_j(M)$, $1 \leq j \leq n$, are arranged in non-increasing order, that is, 
\begin{equation*}
  \lambda_{1}(M) \geq \ldots \geq \lambda_{n-1}(M) \geq \lambda_{n}(M).
\end{equation*}
The multiplicity of $\lambda_j(M)$ as an $M$-eigenvalue is denoted by $m_{M}(\lambda_j)$
The collection of $M$-eigenvalues together with their multiplicities is called the $M$-spectrum, denoted by $\sigma(M)$, i.e., if $ \lambda_{1}(M) \geq \ldots \geq \lambda_{r}(M) $ are $r$-distinct eigenvalues, $\sigma(M)=\{{\lambda_1(M)}^{m_{M}(\lambda_1)}, \ldots, {\lambda_r(M)}^{m_{M}(\lambda_r)}\}.$

The matrix $M:=(m_{ij})$ is a nonnegative matrix if all its elements are greater than or equal to zero, that is, $m_{ij} \geq 0,$ $\forall \,\ 1 \leq i,j \leq n$. $\widetilde{M}$ is the matrix obtained from $M$ by deleting its last row and its last column. The zero matrix, the identity matrix, the all-ones matrix, and the transpose of a matrix $M$ of the appropriate orders are denoted by $\textbf{0}$, $I$, $J$, and $M^T$, respectively. 

For any arbitrary matrices $C$ and $F,$ the direct sum $C \bigoplus F$ is defined as the block diagonal matrix of $C$ and $F$.
The following results and definitions will be used throughout this work.

\begin{definition}[\cite{you}]\label{quotientm}
Let $M$ be a complex matrix of order $n \times n$ described in the following block form
$$M=\left[
 \begin{array}{cccc}
 M_{11}  & M_{12} & \ldots & M_{1t} \\
 M_{21} & M_{22} & \ldots & M_{2t} \\
 \ldots & \ldots & \ldots & \ldots \\
 M_{t1} & M_{t2} & \ldots & M_{tt} \\
 \end{array}
\right],$$
where the blocks $M_{ij}$ are $n_i \times n_j$ matrices for any $1\leq i,j \leq t$ and $n=n_1+\cdots+n_t$. For $1 \leq i, j \leq t$, let $q_{ij}$ be the average row sum of $M_{ij}$, i.e., $q_{ij}$ is the sum of all entries in $M_{ij}$ divided by the number of rows. Then the matrix $\overline{M}:=(q_{ij})$ is called a quotient matrix of $M$. If, in addition, for each pair $i,j$, $M_{ij}$ has a constant row sum, then $\overline{M}$ is called the equitable quotient matrix of $M$. 
\end{definition}  

Basic facts on the equitable quotient matrix are given in the following lemma.

\begin{lemma}[\cite{you}]\label{quotient}
Let $\overline{M}$ be an equitable quotient matrix of $M$ as defined in Definition \ref{quotientm}. Then
\begin{enumerate}[(i)]
\item $\sigma(\overline{M}) \subseteq \sigma(M)$, and
\item $\rho(\overline{M})=\rho(M)$ if $M$ is nonnegative.
\end{enumerate}
\end{lemma}

\begin{proposition}[Corollary 4.3.15, \cite{HoJosecond}] \label{weyli}
Let $A$ and $B$ be Hermitian matrices of order $n$ and let $1\leq j\leq n$. Then
\begin{equation*}
\lambda_{j}\left(  A\right)  +\lambda_{n}\left(  B\right)  \leq\lambda
_{j}\left(  A+B\right)  \leq\lambda_{j}\left(  A\right)  +\lambda_{1}\left(
B\right).
\end{equation*}
In either of these inequalities, equality holds if and only if there exists a nonzero $n$--vector that is an eigenvector corresponding to the three eigenvalues involved. In particular
\begin{equation*}
  \lambda_{1}\left(  A\right)  +\lambda_{n}\left(  B\right)  \leq\lambda
_{1}\left(  A+B\right)  \leq\lambda_{1}\left(  A\right)  +\lambda_{1}\left(
B\right).
\end{equation*}
\end{proposition}

The next corollary is immediate from Proposition \ref{weyli}.

\begin{corollary}[\cite{HoJosecond}]\label{corpsd}
Let $M$ and $N$ be Hermitian matrices of order $n$ and let $1\leq j \leq n$. If $N$ is positive semidefinite, then $\lambda_j(M) \leq \lambda_j(M+N)$.
\end{corollary}

Let $R$ be a matrix of appropriate order whose entries are zero except the entry in the last row and the last column, which is $1$.

\begin{lemma}[\cite{RM}]\label{extra}
For $i,j=1,2, \ldots, m$, let $M_{i}$ be a matrix of order $k_i \times k_i$, with $k_i\geq 2$, and $ \mu _{i,j}$ be arbitrary scalars. Then
\begin{eqnarray*}
&&\left\vert
\begin{array}{ccccc}
M_{1} & \mu_{1,2}R & \cdots & \mu _{1,m-1}R & \mu_{1,m}R \\
\mu _{2,1}R^{T} & M_{2} & \cdots & \cdots & \mu _{2,m}R \\
\mu _{3,1}R^{T} & \mu _{3,2}R^{T} & \ddots & \cdots & \vdots \\
\vdots & \vdots & \vdots & M_{m-1} & \mu _{m-1,m}R \\
\mu _{m,1}R^{T} & \mu _{m,2}R^{T} & \cdots & \mu _{m,m-1}R^{T} & M_{m}%
\end{array}%
\right\vert \\
&=&\left\vert
\begin{array}{ccccc}
\left\vert M_{1}\right\vert & \mu _{1,2}\left\vert \widetilde{M_{2}}\right\vert & \cdots & \mu _{1,m-1}\left\vert \widetilde{M_{m-1}}\right\vert & \mu _{1,m}\left\vert \widetilde{M_{m}}\right\vert \\
\mu _{2,1}\left\vert \widetilde{M_{1}}\right\vert & \left\vert M_{2}\right\vert & \cdots & \cdots & \mu _{2,m}\left\vert \widetilde{M_{m}}\right\vert \\
\mu _{3,1}\left\vert \widetilde{M_{1}}\right\vert & \mu _{3,2}\left\vert \widetilde{M_{2}}\right\vert & \ddots & \cdots & \vdots \\
\vdots & \vdots & \vdots & \left\vert M_{m-1}\right\vert & \mu_{m-1,m}\left\vert \widetilde{M_{m}}\right\vert \\
\mu _{m,1}\left\vert \widetilde{M_{1}}\right\vert & \mu _{m,2}\left\vert\widetilde{M_{2}}\right\vert & \cdots & \mu _{m,m-1}\left\vert \widetilde{M_{m-1}}\right\vert & \left\vert M_{m}\right\vert%
\end{array}%
\right\vert .
\end{eqnarray*}
\end{lemma}

The following result was obtained in \cite{ACPR}, and considers the case $k_i=1$, for some $1 \leq i \leq m$, where it was shown that Lemma \ref{extra} also holds.

\begin{corollary}[\cite{ACPR}]\label{extra1}
If $k_i=1$, for some $1 \leq i \leq m$, then defining $\widetilde{M_{i}}=1$ the equality in Lemma \ref{extra} also holds.
\end{corollary}

\section{Properties of \texorpdfstring{$B_{\alpha}$}{Balpha}-matrices of graphs}\label{prop_B_alpha}

Let $G$ be a graph of order $n$ and let $x := (x_1,x_2,\ldots,x_n)$ be a real vector, whose norm is denoted by $||x||$. From Equation \eqref{identity2}, the quadratic form $\langle B_\alpha(G) x, x\rangle$ can be expressed as:
\begin{equation*}
\langle B_\alpha(G) x, x\rangle = \alpha \sum\limits_{u\in V(G)}d_G(u)x_u^2+(1-2\alpha)\sum\limits_{\{u,v\} \in E(G)}(x_u-x_v)^2
\end{equation*}
Additionally,
\begin{eqnarray}\label{forq2}
    \langle B_\alpha(G) x, x\rangle &=& (1-\alpha)\sum\limits_{u\in V(G)}d_G(u)x_u^2+2(2\alpha-1)\sum\limits_{\{u,v\} \in E(G)}x_ux_v, 
\end{eqnarray}
\begin{eqnarray}\label{forq3}
    \langle B_\alpha(G) x, x\rangle &=& \sum\limits_{\{u,v\} \in E(G)}((1-\alpha) x_u^2+2(2\alpha-1)x_ux_v+ (1-\alpha) x_v^2).
\end{eqnarray}

Since $B_{\alpha}(G)$ is a real symmetric matrix, Rayleigh's principle implies the following result.

\begin{proposition}\label{niki1}
    Let $G$ be a graph of order $n$ and $0 \leq \alpha  \leq 1$. Then 
\begin{equation*}
 \lambda_1(B_{\alpha}(G))=\max_{||x||=1} \langle B_\alpha(G) x, x\rangle  
  \ and \ \lambda_n(B_{\alpha}(G)) =\min_{||x||=1} \langle B_\alpha(G) x, x\rangle.
\end{equation*}
 Moreover, if $x$ is a unit $n$-vector then $\lambda_1(B_{\alpha}(G))= \langle B_\alpha(G) x, x\rangle$ if and only if $x$ is an eigenvector to $\lambda_1(B_{\alpha}(G)),$ and $\lambda_n(B_{\alpha}(G))= \langle B_\alpha(G) x, x\rangle$ if and only if $x$ is an eigenvector to $\lambda_n(B_{\alpha}(G)).$
\end{proposition}

\begin{proposition}\label{nik15}
Let $G$ be a graph with $u,v,w \in V(G)$, such that $\{ u,v \}\in E(G)$ and $\{ u,w \} \notin E(G)$. Let $H$ be the graph obtained from $G$ by deleting the edge $\{u,v \}$ and adding the edge $\{u,w \}$. Suppose that $x:=(x_1,\ldots,x_n)$ is a unit eigenvector to $\lambda_1(B_{\alpha}(G))$, with $x_u>0$. If
\begin{enumerate}[(i)]
    \item $0 \leq \alpha < \frac{1}{2}$ and $x_w \leq x_v <0$, or
    \item $\frac{1}{2} < \alpha \leq 1$ and $0 < x_v \leq x_w$, 
\end{enumerate}
then $\lambda_1(B_{\alpha}(H))>\lambda_1(B_{\alpha}(G))$.
\end{proposition}

\begin{proof}
Let $x:=(x_1,\ldots,x_n)$ be a unit eigenvector to $\lambda_1(B_{\alpha}(G))$, with $x_u>0$. Rayleigh's principle implies that $\lambda_1(B_{\alpha}(H)) \geq x^TB_{\alpha}(H)x$. Thus, 

\begin{equation*}
  \lambda_1(B_{\alpha}(H))-\lambda_1(B_{\alpha}(G)) \geq x^TB_{\alpha}(H)x-x^TB_{\alpha}(G)x.   
\end{equation*}
From Equation (\ref{forq3}), we have
\begin{equation*}
x^TB_{\alpha}(H)x = \sum\limits_{\{i,j\} \in E(H)}((1-\alpha) x_i^2+2(2\alpha-1)x_ix_j+ (1-\alpha) x_j^2),
\end{equation*}
and
\begin{equation*}
x^TB_{\alpha}(G)x = \sum\limits_{\{i,j\} \in E(G)}((1-\alpha) x_i^2+2(2\alpha-1)x_ix_j+ (1-\alpha) x_j^2).
\end{equation*}
So,
\begin{eqnarray*}
x^TB_{\alpha}(H)x-x^TB_{\alpha}(G)x &=& (1-\alpha) x_u^2+2(2\alpha-1)x_ux_w+ (1-\alpha) x_w^2 \\
& \,\ & -((1-\alpha) x_u^2+2(2\alpha-1)x_ux_v+ (1-\alpha) x_v^2),
\end{eqnarray*}
and consequently
\begin{equation}\label{eqdif}
x^TB_{\alpha}(H)x-x^TB_{\alpha}(G)x = 2(2\alpha-1)x_u(x_w-x_v)+ (1-\alpha) (x_w^2-x_v^2).
\end{equation}

If $0 \leq \alpha < \frac{1}{2}$ and $x_w \leq x_v <0$, from Equation (\ref{eqdif}), we have $$x^TB_{\alpha}(H)x \geq x^TB_{\alpha}(G)x.$$
Then
$$\lambda_1(B_{\alpha}(H)) \geq\lambda_1(B_{\alpha}(G)).$$

However, the equality $\lambda_1(B_{\alpha}(H))=\lambda_1(B_{\alpha}(G))$ is not possible. Suppose, towards a contradiction, that $\lambda_1(B_{\alpha}(H))=\lambda_1(B_{\alpha}(G))$. Proposition \ref{niki1} implies that $x$ is an eigenvector to $\lambda_1(B_{\alpha}(H))$. Since $x_w \leq x_v <0$ and $x_u>0$, from Equation (\ref{forq2}) we have
\begin{eqnarray*}
\lambda_1(B_{\alpha}(H))x_{w}&=&(1-\alpha)d_H(w)x_w+(2\alpha-1)\sum\limits_{\{j,w\} \in E(H)}x_j x_w \\
&<& (1-\alpha)d_G(w)x_w+(2\alpha-1)\sum\limits_{\{j,w\} \in E(G)}x_j x_w\\
&=&\lambda_1(B_{\alpha}(G))x_{w}.
\end{eqnarray*}
Then $\lambda_1(B_{\alpha}(H)) > \lambda_1(B_{\alpha}(G))$, which is a contradiction. Thus, the result of case $(i)$ follows.

If $\frac{1}{2} < \alpha \leq 1$ and $0 < x_v \leq x_w$, from Equation (\ref{eqdif}), the proof is similar to case $(i)$.
\end{proof}

\begin{remark}
Let G and H be graphs as in the hypotheses of Proposition \ref{nik15}. If $\alpha=1/2$, we have $B_{\frac{1}{2}}(G)=\frac{1}{2}D(G)$. Then it is immediate that if $\Delta(H)$ is less than $\Delta(G)$ (equal to, greater than, respectively), then $\lambda_1(B_{\alpha}(H))$ is less than $\lambda_1(B_{\alpha}(G))$ (equal to, greater than, respectively).
\end{remark}

In \cite{Samanta2024}, it was proven that if $G$ is connected, then $B_\alpha(G)$ is irreducible for $\alpha \in [0,\frac{1}{2}) \cup (\frac{1}{2},1]$. Furthermore, according to the definition, $B_\alpha$-matrices are nonnegative if and only if $\frac{1}{2} \leq \alpha \leq 1$. However, in \cite{Samanta2024}, the authors conclude that for $\alpha \in [0,1]$, the spectral radius of $B_\alpha(G)$ equals its largest eigenvalue, as can be seen in Theorem \ref{sradius}.

\begin{theorem}[\cite{Samanta2024}]\label{sradius}
For any $\alpha \in [0,1]$, the spectral radius of $B_\alpha(G)$ of a connected graph $G$ is $\lambda_1(B_\alpha(G))$.
\end{theorem}

Let $\tilde{G}$ be a subgraph of $G$ obtained by removing an edge. The next result analyzes the behavior of the eigenvalues of $B_{\alpha}(G)$ and $B_{\alpha}(\tilde{G}).$

\begin{proposition}\label{delete}
Let $G$ be a graph of order $n$ with at least one edge, and let $\tilde{G}$ be the graph obtained from $G$ by removing an edge $e \in E(G)$. Then,
\begin{enumerate}[(i)]
    \item if $0 \leq \alpha \leq \frac{2}{3}$, we have 
    $$\lambda_j(B_\alpha(G)) \geq \lambda_j (B_\alpha (\tilde{G}))$$ for $j=1,\ldots,n$.
 \item if $\frac{1}{2} < \alpha \leq 1$, we have 
$$\lambda_1(B_\alpha(G)) > \lambda_1 (B_\alpha (\tilde{G})).$$
\end{enumerate}
\end{proposition}
\vspace{-15mm}
\begin{proof}
Without loss of generality, assume that $e=\{v_1,v_2\}$. Let $\tilde{G}$ be the graph obtained from $G$ by removing this edge. Then, 
$$B_\alpha(G)=B_\alpha(\tilde{G})+M$$
where
$$M= \left[\begin{array}{ccccc}
        1-\alpha & 2\alpha-1 & 0 & \cdots & 0\\
        2\alpha-1 & 1-\alpha & 0 & \cdots & 0\\
        0 & 0 & \cdots & \cdots & 0 \\
        \vdots & \vdots & \vdots & \vdots & \vdots\\
        0 & 0 & \cdots & \cdots & 0\\
\end{array} \right].
$$
The eigenvalues of $M$ are $\alpha$, $2-3\alpha$, and $0$ with multiplicity $n-2$. Hence, if $0 \leq  \ \alpha \leq \frac{2}{3}$, $M$ is a positive semidefinite matrix, then by Corollary~\ref{corpsd} the result in (i) follows. The result (ii) is an immediate consequence of the Perron-Frobenius theory of nonnegative irreducible matrices.
\end{proof}

\begin{remark}
For $j \ne 1$ and $\dfrac{2}{3} < \alpha \leq 1$, the inequality in Proposition~\ref{delete}, item (i) given by $\lambda_j(B_\alpha(G)) \geq \lambda_j (B_\alpha (\tilde{G}))$ is not satisfied. For example, consider the graph $G$ depicted in Figure~\ref{fig:graph}, and let $\tilde{G}$ be the graph obtained by removing the edge $e=\{v_3, v_5 \}$ of $G$. As illustrated in Figure~\ref{fig:behavior_delete}, the $B_{\alpha}$-eigenvalues of $G$ and $\tilde{G}$ do not satisfy the aforementioned inequality.

\begin{figure}[H]
    \centering
    \includegraphics[width=\linewidth]{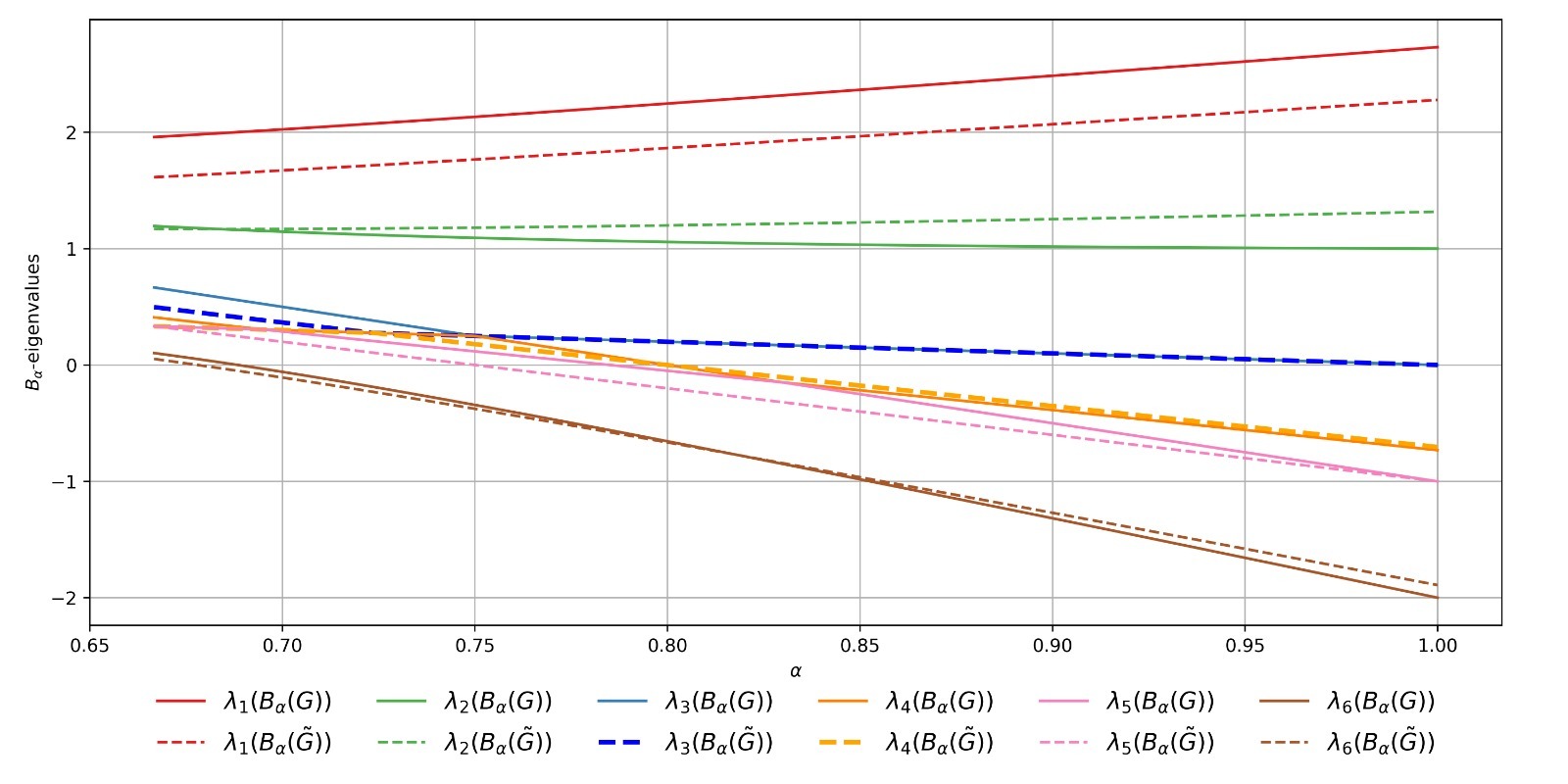}
    \caption{Behavior the $B_{\alpha}$-eigenvalues of $G$ and $\tilde{G}$ for $\alpha \in \left(\dfrac{2}{3}, 1\right]$}.
    \label{fig:behavior_delete}
\end{figure}
\end{remark}

Proposition \ref{convex} shows that the eigenvalues $\lambda_1(B_{\alpha}(G))$ and $\lambda_n(B_{\alpha}(G))$ are convex and concave functions in $\alpha$, respectively.

\begin{proposition}\label{convex}
Let $G$ be a graph of order $n$ and $0 \leq \alpha \leq 1$. Then,
\begin{enumerate}[(i)]
  \item $\lambda_1(B_{\alpha}(G))$ is convex in $\alpha$
  \item $\lambda_{n}(B_{\alpha}(G))$ is concave in $\alpha$.
\end{enumerate}
\end{proposition}
\begin{proof}
Let $\alpha_1, \alpha_2 \in [0,1]$. If we consider,
\begin{equation*}
\alpha= (1-t) \alpha_1 + t \alpha_2 \,\ , \,\ 0\leq t \leq 1,
\end{equation*}
then
\begin{eqnarray*}
B_{\alpha}(G)&=& B_{(1-t)\alpha_1 + t \alpha_2}(G) \\
&=& ((1-t)\alpha_1 + t \alpha_2 )A(G) + (1-((1-t) \alpha_1 + t \alpha_2))L(G) \\
&=& ((1-t)\alpha_1 + t \alpha_2 )A(G) + (1-((1-t) \alpha_1 + t \alpha_2))L(G) + tL(G) -tL(G) \\
&=& (1-t)(\alpha_1A(G) +(1-\alpha_1)L(G)) +t(\alpha_2A(G) + (1-\alpha_2)L(G)) \\
&=& (1-t) B_{\alpha_1}(G) +tB_{\alpha_2}(G).
\end{eqnarray*}
So, $B_{\alpha}(G)=(1-t) B_{\alpha_1}(G) +tB_{\alpha_2}(G)$.
Considering the function $\lambda_k(B_{\alpha}(G))$ and Proposition \ref{weyli}, we obtain
\begin{equation*}
(1-t) \lambda_k(B_{\alpha_1}(G)) +t \lambda_n(B_{\alpha_2}(G)) \leq \lambda_k(B_{\alpha}(G)) \leq (1-t) \lambda_k(B_{\alpha_1}(G)) +t \lambda_1(B_{\alpha_2}(G)).
\end{equation*}
If $k=1$, we have that
\begin{equation*}
\lambda_1(B_{\alpha}(G)) \leq (1-t) \lambda_1(B_{\alpha_1}(G)) +t \lambda_1(B_{\alpha_2}(G)),
\end{equation*}
then, it follows the convexity of $\lambda_1(B_{\alpha}(G))$ in $\alpha$. Similarly, if $k=n$, we have that
\begin{equation*}
(1-t) \lambda_n(B_{\alpha_1}(G)) +t \lambda_n(B_{\alpha_2}(G)) \leq \lambda_n(B_{\alpha}(G)),
\end{equation*}
thus, it follows the concavity of $\lambda_n(B_{\alpha}(G))$ in $\alpha$.
\end{proof}

We conclude this section by showing a relationship between the $B_{\alpha}$-eigenvalues and $A_\alpha$-eigenvalues of $G$. Consider the following identity:
\begin{eqnarray*}
B_{\alpha}(G) &=& \alpha A(G) + (1-\alpha)L(G) \\
%&=& \alpha A(G) + (1-\alpha)(D(G)-A(G)) \\
%&=& \alpha A(G) + (1-\alpha)D(G)-(1-\alpha) A(G) \\
&=& - (\alpha D(G)+(1-\alpha) A(G)) + \alpha A(G) + D(G).
\end{eqnarray*}
Therefore, 
\begin{equation}\label{identitynew}
B_{\alpha}(G)=-A_{\alpha}(G)+M_{\alpha},    
\end{equation}
where $M_{\alpha}=\alpha A(G) + D(G)$. Note that the entries of $M_{\alpha}:=(m_{i,j})$ are given by $m_{i,i}=d_G(v_i)$ in the $i$-th diagonal entry, $m_{i,j}=\alpha$ if the vertices $v_i$ and $v_j$ are adjacent, and $0$ otherwise. From Gershgorin's theorem, it follows that $M_{\alpha}$ is a positive semidefinite matrix and applying Corollary \ref{corpsd} to Equation (\ref{identitynew}), we obtain the following result.  

\begin{proposition}
Let $G$ be a graph of order $n$ and $0\leq \alpha \leq 1$. Then, 
\begin{equation*}
\lambda_{k}(B_{\alpha}(G)) \geq -\lambda_k(A_{\alpha}(G)),    
\end{equation*}
for $k=1,\ldots,n$.
\end{proposition}

\section{Multiplicity of \texorpdfstring{$B_\alpha$}{Balpha}-eigenvalues of graphs}\label{multiplicity}

In this section, we present some results involving the $B_{\alpha}$-eigenvalues of $G$ with their multiplicities when some sets of vertices satisfy certain conditions. More precisely, we consider conditions when the sets of vertices are independent sets, cliques, false twins, or true twins. Moreover, we present the study of $1-\alpha$ as $B_{\alpha}$-eigenvalue of graphs with pendant vertices, which will allow us to unify the study of some eigenvalues in the literature on Spectral Graph Theory. We begin by presenting lower bounds for the multiplicity of some $B_{\alpha}$-eigenvalues of a graph when we consider conditions on its independent sets or cliques.

\begin{proposition}\label{indep1}
    Let $G$ be a graph of order $n$ and $0 \leq \alpha \leq 1$. If $S \subseteq V(G)$ is an independent set, such that $N_G(x)=N_G(y)$ for any $x,y \in S$, then  
    \begin{equation*}
        m_{B_{\alpha}(G)}((1-\alpha)d) \geq |S| - 1,
    \end{equation*}
    where $d=|N_G(x)|$ for any $x\in S$. 
\end{proposition}

\begin{proof}
Let $S$ be an independent set of $G$. Consider
$\textbf{e}_i$ be the $n$-vector with $1$ in the $i$-th entry and $0$ elsewhere. Making $\textbf{x}_i=\textbf{e}_1-\textbf{e}_{i+1}$, for $i=1,2,\ldots,|S|-1$; and labeling the vertices in $S$ with $\{1,2,\ldots,|S|\}$, and the remaining vertices with $\{ |S|+1, \ldots, n \}$, we have 
\begin{equation*}
B_{\alpha}(G)\textbf{x}_i=( (1-\alpha)d, 0, \ldots, 0, \underbrace{-(1-\alpha)d}_{\text{$i$-th \,\ entry}}, 0, \ldots, 0 )=(1-\alpha)d \textbf{x}_i.    
\end{equation*}
Clearly, $\textbf{x}_1, \textbf{x}_2, \ldots, \textbf{x}_{|S|-1}$ are linearly independent. Thus $(1-\alpha)d$ is an $B_{\alpha}$-eigenvalue with multiplicity at least $|S|-1$.
\end{proof}

As a consequence of  Proposition \ref{indep1}, we have the following result.

\begin{corollary}\label{indepgen}
 Let $G$ be a graph and $0 \leq \alpha \leq 1$. If, for $1 \leq i \leq t$, $S_i \subseteq V(G)$ is an independent set, such that:
\begin{enumerate}[(i)]
\item $N_G(x)=N_G(y)$, for any $x,y \in S_i$, 
\item $N_G(u) \cap N_G(v)=\emptyset$, for $u \in S_i$  and $v \in S_j$, with $i \neq j$, and
\item $d_{G}(u)=d_{G}(v)$, for $u \in S_i$  and $v \in S_j$, with $i \neq j$, 
\end{enumerate}
then 
\begin{equation*}
        m_{B_{\alpha}(G)}((1-\alpha)d_i) \geq \displaystyle\sum_{i=1}^{t}|S_i|-t,
    \end{equation*}
    where $d_i=d_{G}(x)$ for any $x\in S_i$.    
\end{corollary}

\begin{proposition}\label{clique}
Let $G$ be a graph, and $0 \leq \alpha \leq 1$. If $W \subseteq V(G)$ is a clique such that $N_G(x)-W=N_G(y)-W$ for any $x,y \in W$, then \begin{equation*}
        m_{B_{\alpha}(G)}(d+1 -\alpha (d+2)) \geq |W| - 1,
    \end{equation*}
    where $d=|N_G(x)|$ for any $x\in W$. 
\end{proposition}

\begin{proof}
The proof is similar to Proposition \ref{indep1}. If $W$ is a clique, for $1 \leq i \leq |W|-1$, we have 
\begin{equation*}
B_{\alpha}(G)\textbf{x}_i=( d+1 -\alpha (d+2), 0, \ldots, 0, \underbrace{-(d+1 -\alpha (d+2))}_{\text{$i$-th \,\ entry}}, 0, \ldots, 0 )=(d+1 -\alpha (d+2))\textbf{x}_i,   
\end{equation*}
where $\textbf{x}_i=e_1-e_{i+1}$. 
Thus, $d+1 -\alpha (d+2)$ is an $B_{\alpha}$-eigenvalue with multiplicity at least $|W|-1$.
\end{proof}

 By repeated applications of Proposition \ref{clique}, the following result follows immediately.

\begin{corollary}\label{cliquegen}
 Let $G$ be a graph, and $0 \leq \alpha \leq 1$. If, for $1 \leq i \leq t$, $W_i \subseteq V(G)$ is a clique, such that:
\begin{enumerate}[(i)]
\item $N_G(x)-W_i=N_G(y)-W_i$, for any $x,y \in W_i$, 
\item $N_G(u) \cap N_G(v)=\emptyset$, for $u \in W_i$  and $v \in W_j$, with $i \neq j$, and
\item $d_{G}(u)=d_{G}(v)$, for $u \in W_i$  and $v \in W_j$, with $i \neq j$, 
\end{enumerate}
then 
\begin{equation*}
        m_{B_{\alpha}(G)}(d_i+1 -\alpha (d_i+2)) \geq \displaystyle\sum_{i=1}^{t}|W_i|-t,
    \end{equation*}
    where $d_i=d_{G}(x)$ for any $x\in W_i$.    
\end{corollary}

Two distinct vertices $u$ and $v$ in $G$, are called true twins if $N_{G}[u] = N_{G}[v]$, and are called false twins if $N_{G}(u) = N_{G}(v)$ and $u$ is nonadjacent to $v$. The next result relates the definition of twin vertices to the results presented in Propositions \ref{indep1}  and \ref{clique}.

\begin{corollary}\label{twinsvertices}
   Let $G$ be a graph, and $0 \leq \alpha \leq 1$. Let $V^T_i$ and $V^F_j$ be sets of pairs of true twins and false twins vertices in $V(G)$, with $1 \leq i \leq t$ and $1 \leq j \leq k$, respectively.

\begin{enumerate}[(i)]
\item If $w_{i}$ and $w_{ip}$ are true twins in $V^T_i$, with $1 \leq i \leq t$, $1 \leq p \leq |V^T_i|-1$, and $V^T_i$ satisfies the conditions in Corollary \ref{cliquegen}, then
\begin{equation*}
        m_{B_{\alpha}(G)}(d_{G}(w_i)+1 -\alpha (d_{G}(w_i)+2)) \geq \displaystyle\sum_{i=1}^{t}|V^T_i|-t.
    \end{equation*}
\item If $u_{j}$ and $u_{jp}$ are false twins in $V^F_j$, with $1 \leq j \leq k$, $1 \leq p \leq |V^F_j|-1$, and $V^F_j$ satisfies the conditions in Corollary \ref{indepgen}, then
\begin{equation*}
        m_{B_{\alpha}(G)}((1-\alpha)d_{G}(u_j)) \geq \displaystyle\sum_{j=1}^{k}|V^F_j|-k.
    \end{equation*} 
\end{enumerate} 

\end{corollary}

\begin{example}
    
The Figure \ref{fig:twins} shows a graph $G$ with sets of pairs of twin vertices given by $V^F_1= \{u_1, u_{11}, u_{12}, u_{13}, u_{14}\}$, $V^F_2= \{u_2, u_{21}, u_{22}\}$, $V^F_3=\{ u_3, u_{31}, u_{32} \}$, $V^T_1=\{ w_1, w_{11} \}$, and $V^T_2=\{ w_2, w_{21}, w_{22} \}$. 

\begin{figure}[H]
    \centering
    %\includegraphics[width=0.6\linewidth]{twinvertices.pdf}
    %\linewidth]{twins}
    \begin{tikzpicture}[
  vtx/.style={circle,fill=black,inner sep=2.0pt},
  lbl/.style={font=\scriptsize,inner sep=1pt},
  edge/.style={line width=.6pt},
  pale/.style={draw=none,fill=black!10,rounded corners=2pt}
]

% --- ESQUERDA ---------------------------------------------------------------
\node[vtx] (u21) at (0,0) {};
\node[vtx] (u22) at (0,-1.2) {};
\node[vtx] (u2)  at (0, 1.2) {};
\node[lbl,anchor=west] at ($(u2)+(-0.1,0.2)$) {$u_{2}$};
\node[lbl,anchor=west] at ($(u21)+(-0.1,0.2)$) {$u_{21}$};
\node[lbl,anchor=west] at ($(u22)+(0.05,-0.2)$) {$u_{22}$};

\node[vtx] (a) at (1,0) {};

\node[vtx] (u32) at (2,-1.2) {};
\node[vtx] (u3)  at (2, 1.2) {};
\node[vtx] (u31) at (2, 0) {};
\node[lbl,anchor=west] at ($(u3)+(0.05,0.12)$)  {$u_{3}$};
\node[lbl,anchor=west] at ($(u31)+(0.05,0.12)$) {$u_{31}$};
\node[lbl,anchor=west] at ($(u32)+(0.07,0)$) {$u_{32}$};

% arestas (não orientadas)
\draw[edge] (u2)--(a)--(u22) (u21)--(a);
\draw[edge] (a)--(u3) (a)--(u31) (a)--(u32);
\draw[edge] (u3)--(a) (u32)--(a);

% conector para o bloco do meio
\node[vtx] (b) at (3,0) {};
\draw[edge] (u31)--(b) (u3)--(b) (u32)--(b);

% --- MEIO -------------------------------------------------------------------
\node[vtx] (w2)  at (3.7, 1.2) {};
\node[vtx] (w22) at (3.7,-1.2) {};
\node[vtx] (w21) at (4.1, 0) {};
\node[vtx] (c)   at (5, 0) {};

\node[lbl,anchor=south] at ($(w2)+(0.3,- 0.1)$)  {$w_{2}$};
\node[lbl,anchor=north] at ($(w22)+(0.4,+ 0.1)$) {$w_{22}$};
\node[lbl,anchor=west]  at ($(w21)+(0.06,0.2)$) {$w_{21}$};

\draw[edge] (b)--(w2) (b)--(w21) (b)--(w22);
\draw[edge] (w2)--(c) (w22)--(c) (w21)--(c);
\draw[edge] (w2)--(w21)--(w22) (w2)--(w22);

% --- DIREITA ----------------------------------------------------------------
\node[vtx] (w1)  at (6, 1.2) {};
\node[vtx] (w11) at (6,-1.2) {};
\node[vtx] (s)   at (7, 0)   {};

\node[lbl,anchor=south] at ($(w1) + (0.3, -0.1)$)  {$w_{1}$};
\node[lbl,anchor=north] at ($(w11) + (0.4, 0.2)$) {$w_{11}$};

% “barra” vertical sob w1 (estético)
\draw[edge] (w1) -- (w11);

\draw[edge] (c)--(w1) (c)--(w11);
\draw[edge] (w1)--(s) (w11)--(s);

\node[vtx] (U1)  at (8.0, 1.2) {};
\node[vtx] (U11) at (8.0, 0.6) {};
\node[vtx] (U12) at (8.0, 0.0) {};
\node[vtx] (U13) at (8.0,-0.6) {};
\node[vtx] (U14) at (8.0,-1.2) {};

\draw[edge] (s)--(U1) (s)--(U11) (s)--(U12) (s)--(U13) (s)--(U14);

\node[lbl,anchor=west] at ($(U1)+(0.06,0.12)$)  {$u_{1}$};
\node[lbl,anchor=west] at ($(U11)+(0.06,0.12)$) {$u_{11}$};
\node[lbl,anchor=west] at ($(U12)+(0.06,0.12)$) {$u_{12}$};
\node[lbl,anchor=west] at ($(U13)+(0.06,0.12)$) {$u_{13}$};
\node[lbl,anchor=west] at ($(U14)+(0.06,0.12)$) {$u_{14}$};

\end{tikzpicture}
    \caption{Graph $G$.}
    \label{fig:twins}
\end{figure}

From Corollary \ref{twinsvertices}, we have that 
\begin{equation*}
        m_{B_{\alpha}(G)}(1-\alpha) \geq \displaystyle\sum_{j=1}^{2}|V^F_j|-2=|V^F_1|+|V^F_2|-2=5+3-2=6,
    \end{equation*}

\begin{equation*}
        m_{B_{\alpha}(G)}(2(1-\alpha)) \geq \displaystyle |V^F_3|-1=3-1=2,
    \end{equation*} 
\begin{equation*}
        m_{B_{\alpha}(G)}(4 -5\alpha) \geq \displaystyle|V^T_1|-1=2-1=1,
    \end{equation*}
and,
\begin{equation*}
        m_{B_{\alpha}(G)}(5 -6\alpha) \geq \displaystyle|V^T_2|-1=3-1=2.
    \end{equation*}

\end{example}

\subsection{Multiplicity of \texorpdfstring{$1-\alpha$}{1-alpha} as \texorpdfstring{$B_{\alpha}$}{Balpha}-eigenvalue of G}

Motivated by the study of $\alpha$ as an $A_{\alpha}$-eigenvalue of graphs with pendant vertices in \cite{multiplicity}, in this subsection, we investigate $1-\alpha$ as an $B_{\alpha}$-eigenvalue of graphs with pendant vertices. This will allow us to unify the study of the multiplicity of $1$ as a Laplacian and signless Laplacian eigenvalue, and the nullity of a graph \cite{ACPR, barik,nul1,applsgt,faria,nul4}. For this purpose, consider the following problem:

\begin{problem}\label{prob1pendant}

When does equality occur for the inequality obtained in Corollary 3?. That is, studying the conditions for it to occur:

\begin{equation*}
        m_{B_{\alpha}(G)}((1-\alpha)d_i) = \displaystyle\sum_{i=1}^{t}|S_i|-t.
    \end{equation*}
\end{problem}

If we assume that the independent sets $S_i$, for $1 \leq i \leq t$, satisfy the conditions of Corollary \ref{indepgen} and that $d_i=d_{G}(x)=1$ for any $x\in S_i$, we obtain the lower bound

\begin{equation}\label{ineqindep}
        m_{B_{\alpha}(G)}(1-\alpha) \geq \displaystyle\sum_{i=1}^{t}|S_i|-t.
\end{equation}

Hereafter, a vertex adjacent to a pendant vertex is called a quasi-pendant vertex, and a vertex of degree at least $2$ is called an internal vertex. The number of pendant vertices, quasi-pendant vertices, and internal vertices are denoted by $p(G), q(G)$, and $r(G) $, respectively. 

Since the independent sets $S_i$ are considered as sets of vertices of degree 1, it follows that $\sum_{i=1}^{t}|S_i|=p(G)$. Moreover, the common neighbors of such vertices in $S_i$ are quasi-pendant vertices, yielding $t=q(G)$. Therefore, the lower bound (\ref{ineqindep}) can be expressed as follows

\begin{equation}\label{ineq.eig}
        m_{B_{\alpha}(G)}(1-\alpha) \geq \displaystyle p(G)-q(G).
\end{equation}

To investigate equality in (\ref{ineq.eig}), consider the following sets:
\begin{align*}
    V_P(G) &= \left\{ v\in V\left( G\right) :v\text{ is a pendant vertex}\right\},\\
    V_Q(G)&= \left\{ v\in V\left( G\right) :v\text{ is a quasi-pendant vertex}\right\} \text{ and}\\
    V_C(G)&= V\left( G\right) \setminus \left( V_{P}(G) \cup V_{Q}(G)\right).
\end{align*}
It is easy to see that $V_P(G)$ is an independent set and $V_C(G)$ is the set of the internal vertices of $G$ which are not quasi-pendant and consequently $\left\vert V_C(G)\right\vert=n-p\left( G\right) -q\left( G\right)$. Throughout this paper, unless otherwise stated, $v_1,v_2,\ldots,v_{r(G)}$ are the internal vertices of $G$. We label the vertices of $G$ with the numbers $1,2,\ldots,n$, starting with the vertices of the stars $K_{1,s_1},\ldots, K_{1,s_{q(G)}}$, at each star the vertices are labeled beginning with the pendant vertices, and finishing with the internal vertices, which are not quasi-pendants. This labeling of the vertices of $G$ is called the global labeling of the vertices of $G$. For $v_i$ and $v_j$ internal vertices of $G$, let $\varepsilon_{i,j}=1$ if $v_{i}$ is adjacent to $v_{j}$ and $\varepsilon_{i,j}=0$, otherwise. 

\begin{example} Let $G$ be a graph as in Figure~\ref{fig:pendantvertices}. Considering the global labeling, we have that $V_P(G)=\{1,3,5,6,8,9,11,12,13\},$ $V_Q(G) = \{2,4,7,10,14\}$ and $V_C(G)=\{15,16,17,18\}$. For instance, we have that $\varepsilon_{2,15}=1$ and $\varepsilon_{16,17}=0$.

\begin{figure}[H]
\centering
\begin{tikzpicture}[
    vertex/.style={circle, draw, fill=black, minimum size=5pt, inner sep=0pt},
    every label/.style={font=\scriptsize},
    scale=1
]

% --- Vértices da linha de base ---
\node[vertex,label=below:1]  (1)  at (0,0)  {};
\node[vertex,label=below:2]  (2)  at (1,0)  {};
\node[vertex,label=below:15] (15) at (2,0)  {};
\node[vertex,label=below:16] (16) at (3,0)  {};
\node[vertex,label=below:18] (18) at (4,0)  {};
\node[vertex,label=below:17] (17) at (5,0)  {};
\node[vertex,label=below:4]  (4)  at (6,0)  {};
\node[vertex,label=below:3]  (3)  at (7,0)  {};

% --- Vértices superiores à esquerda ---
\node[vertex,label={[label distance=0.3pt]above right:14}] (14) at (3,1.2) {};
\node[vertex,label=above:12] (12) at (2.4,2.2) {};
\node[vertex,label=above:11] (11) at (2.8,2.5) {};
\node[vertex,label=above:13] (13) at (2.2,1.5) {};

% --- Vértices superiores à direita ---
\node[vertex,label=right:10] (10) at (4,1.2) {};
\node[vertex,label=above:9]  (9)  at (3.7,2.2) {};
\node[vertex,label=above:8]  (8)  at (4.3,2.2) {};

% --- Vértices do “galho” direito ---
\node[vertex,label=left:7] (7) at (5.4,0.8) {};
\node[vertex,label=right:6] (6) at (5.8,1.2) {};
\node[vertex,label=right:5] (5) at (6.0,0.4) {};

% --- Quadrilátero central ---
\draw (14) -- (10) -- (18) -- (16) -- cycle;

% --- Ligações da linha de base ---
\draw (1) -- (2) -- (15) -- (16) -- (18) -- (17) -- (4) -- (3);

% --- Ligações superiores ---
\draw (14) -- (12);
\draw (14) -- (11);
\draw (14) -- (13);
\draw (10) -- (9);
\draw (10) -- (8);

% --- Ligações extras do quadrilátero ---
\draw (16) -- (14);
\draw (10) -- (18);
\draw (10) -- (16);

% --- Galho direito ---
\draw (17) -- (7) -- (6);
\draw (7) -- (5);

\end{tikzpicture}
\caption{Graph $G$ with $p(G)=9, q(G)= 5$ and $r(G)=9.$}
\label{fig:pendantvertices}
\end{figure}
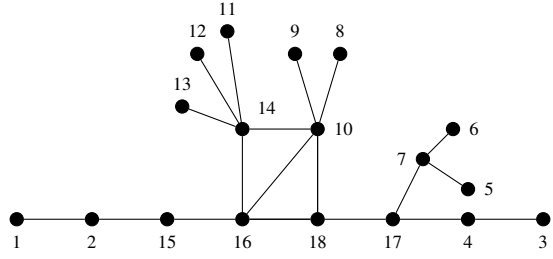
\end{example}

Let $G[F]=G[V_Q(G) \cup V_C(G)]$ be the subgraph induced by the set of internal vertices $F=\{ v_i  :  1 \leq i \leq r(G) \}$. Notice that, since the internal vertices of $G$ are also denoted by $v_1, \dots, v_{q(G)}, \dots, v_{r(G)}$, for $i=1, \dots, q(G)$ each of these vertices $v_i$ corresponds, in the global labeling of the vertices of $G$, to the vertex $\sum_{i=1}^{q(G)}{s_i} + q(G)$ and the vertices $v_{q(G)+1}, \dots, v_{r(G)}$ correspond to the vertices $q(G)+p(G)+1, \dots, n$, respectively.
 
Now, we present Lemma~\ref{lemS} that we will use in the other results in this subsection.

\begin{lemma} \label{lemS}
Consider the square matrix of order $s+1$
$$
E(\alpha)=\left[\begin{array}{ccccc}
                1 - \alpha &     0    & \ldots &     0    & 2\alpha - 1\\
                     0 &  1 -\alpha  & \ldots &     0    & 2\alpha - 1\\
                \vdots &  \vdots  & \ddots &  \vdots  & \vdots \\
                     0 &     0    & \ldots &  1-\alpha  & 2\alpha - 1 \\
              2\alpha -1 & 2\alpha -1 & \ldots & 2\alpha -1 & (1-\alpha) d \\
               \end{array}\right],
$$
 such that $d$ and $\alpha$ are real numbers. Then the characteristic polynomial of $E(\alpha)$ is
            \begin{equation*}
             |xI-E(\alpha)|=(x - (1-\alpha))^{s-1}((x - (1-\alpha) d) (x-(1-\alpha))-s(2\alpha-1)^{2}).
            \end{equation*}
\end{lemma}
\begin{proof}
 
The result follows directly from calculating the determinant using Laplace's method.
\end{proof}

We divide this subsection into two parts. First, we consider graphs in which each internal vertex is quasi-pendant. In the second part, we consider graphs with internal vertices that are not quasi-pendant.

\subsubsection{Graphs in which each internal vertex is quasi-pendant}

In this subsection, we consider a graph \( G \) such that \( r(G) = q(G) \). Thus, there are stars \( K_{1,s_1}, \ldots, K_{1,s_{q(G)}} \), where \( s_i \geq 1 \) for all \( 1 \leq i \leq q(G) \), such that \( G \) is obtained by identifying the root of each \( K_{1,s_i} \) with the $i$-th vertex $v_i$ of the subgraph induced by the quasi-pendant vertices. Hence,
\[
s_1 + s_2 + \cdots + s_{q(G)} + q(G) = p(G) + q(G) = n,
\]
and we denote this graph by \( G(s_1, s_2, \ldots, s_{q(G)}) \). 

\begin{example}
To show the relationship between the labeling $v_i$ of the internal vertices and the global labeling, we display a graph \( G \cong  G(2,3,2,2,1,4,2,1) \), and the subgraph induced \( G[F] \) with $F = \{v_{1}, v_{2}, v_{3}, v_{4}, v_{5}, v_{6}, v_{7}, v_{8}\}$, in the Figure~\ref{fig:G(232)} and Figure~\ref{fig:G[F]}, respectively. Note that, according to the general labeling, we have $v_{1}=3$, $v_{2}=7$, $v_{3}=10$, $v_{4}=13$, $v_{5}=15$, $v_{6}=20$, $v_{7}=23$ and $v_{8}=25$.

\begin{figure}[H]\label{fig:two_graphs}
\centering

% ---- estilos globais p/ os dois tikz ----
\begin{subfigure}[b]{0.4\textwidth}
\centering
\begin{tikzpicture}[
  vertex/.style={circle,draw,fill=black,minimum size=5pt,inner sep=0pt},
  every label/.style={font=\scriptsize, label distance=0.1pt},
  scale=0.7
]
% --- quadrado central ---
\node[vertex,label=below left:25] (25) at (0,2) {};
\node[vertex,label= below left:20] (20) at (2,2) {};
\node[vertex,label=above left:3]   (3)  at (0,0) {};
\node[vertex,label=below right:15] (15) at (2,0) {};
\draw (25)--(20)--(15)--(3)--cycle;

% --- “telhado” ---
\node[vertex,label=left:23] (23) at (1,3) {};
\draw (25)--(23)--(20);

\node[vertex,label=above:22] (22) at (0.75,3.8) {};
\node[vertex,label=above:21] (21) at (1.25,3.8) {};
\draw (23)--(22) (23)--(21);

% --- folhas à esquerda/topo ---
\node[vertex,label=left:24] (24) at (-0.8,2.1) {};
\draw (25)--(24);

% --- estrela no 20 ---
\node[vertex,label=right:16] (16) at (3,2) {};
\node[vertex,label=right:17] (17) at (3.2,2.6) {};
\node[vertex,label=above:18] (18) at (2.6,2.8) {};
\node[vertex,label=above:19] (19) at (2,2.6) {};
\draw (20)--(16) (20)--(17) (20)--(18) (20)--(19);

% --- caminho à direita ---
\node[vertex,label=above:14] (14) at (2.6,0.8) {};
\draw (15)--(14);

\node[vertex,label=below:13] (13) at (3.5,0) {};
\node[vertex,label=below:10] (10) at (5,0)   {};
\draw (15)--(13)--(10);

\node[vertex,label=right:9] (9) at (5.7,0.6) {};
\node[vertex,label=right:8] (8) at (5.7,-0.5) {};
\draw (10)--(9) (10)--(8);

\node[vertex,label=above:12] (12) at (3.5,1.0) {};
\node[vertex,label=above:11] (11) at (4.2,1.0) {};
\draw (13)--(12) (13)--(11);

% --- lado esquerdo inferior + “estrela” no 7 ---
\node[vertex,label=left:1] (1) at (-1,0) {};
\node[vertex,label=left:2] (2) at (-0.6,-0.8) {};
\draw (1)--(3) (2)--(3) (25)--(3);

\node[vertex,label=left:7] (7) at (1,-1.2) {};
\draw (3)--(7)--(15);

\node[vertex,label=left:4] (4) at (0.5,-1.9) {};
\node[vertex,label=below:5] (5) at (1.0,-2.2) {};
\node[vertex,label=right:6] (6) at (1.5,-1.9) {};
\draw (7)--(4) (7)--(5) (7)--(6);
\end{tikzpicture}
\caption{A graph \( G \cong  G(2,3,2,2,1,4,2,1) \).}
\label{fig:G(232)}
\end{subfigure}
\hspace{0.15\textwidth}
\begin{subfigure}[b]{0.4\textwidth}
\centering
\begin{tikzpicture}[
  vertex/.style={circle,draw,fill=black,minimum size=5pt,inner sep=0pt},
  every label/.style={font=\scriptsize, label distance=0.1pt},
  scale=0.7
]
% --- quadrado central ---
\node[vertex,label=left:25] (25) at (0,2) {};
\node[vertex,label=right:20] (20) at (2,2) {};
\node[vertex,label=left:3]   (3)  at (0,0) {};
\node[vertex,label=below right:15] (15) at (2,0) {};
\draw (25)--(20)--(15)--(3)--(25);

% --- “telhado” ---
\node[vertex,label=above:23] (23) at (1,3) {};
\draw (25)--(23)--(20);

% --- triângulo inferior ---
\node[vertex,label=below:7] (7) at (1,-1.2) {};
\draw (3)--(7)--(15);

% --- caminho à direita ---
\node[vertex,label=below:13] (13) at (3.5,0) {};
\node[vertex,label=below:10] (10) at (5,0) {};
\draw (15)--(13)--(10);

\end{tikzpicture}
\caption{$G[F]$.}
\label{fig:G[F]}
\end{subfigure}

\caption{Relationship between the labeling of the internal vertices and the global labeling.}
\end{figure}
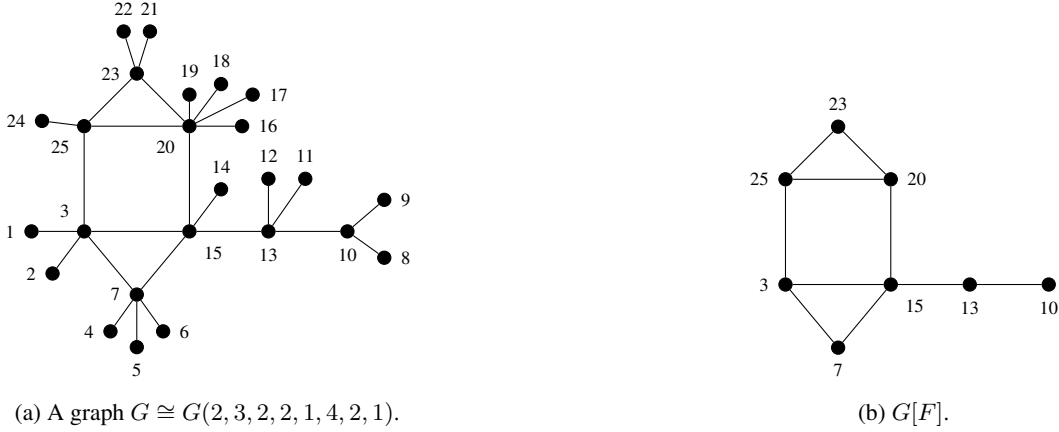

\end{example}

\begin{corollary}\label{aux}
Let \( G(s_1, s_2, \ldots, s_{q(G)}) \) be a graph of order $n$, such that \( s_i \geq 1 \), for \( 1 \leq i \leq q(G) \). For $0 \leq \alpha \leq 1$, define the matrices \( E_i(\alpha) \) of order \( (s_i + 1) \times (s_i + 1) \), and \( C_i(\alpha) \) by:
\begin{equation}\label{si} 
E_i(\alpha) = \left[\begin{array}{ccccc}
1 - \alpha & 0 & \ldots & 0 & 2\alpha - 1 \\
0 & 1 - \alpha & \ldots & 0 & 2\alpha - 1 \\
\vdots & \vdots & \ddots & \vdots & \vdots \\
0 & 0 & \ldots & 1 - \alpha & 2\alpha - 1 \\
2\alpha - 1 & 2\alpha - 1 & \ldots & 2\alpha - 1 & (1 - \alpha) d(v_i) \\
\end{array}\right],
\end{equation}
and
\begin{equation*}
C_i(\alpha) = \left[
\begin{array}{cc}
1 - \alpha & (2\alpha - 1)\sqrt{s_i} \\
(2\alpha - 1)\sqrt{s_i} & (1 - \alpha) d(v_i) \\
\end{array}
\right],
\end{equation*}
where \( v_i \) is a quasi-pendant vertex and \( d(v_i) \) is its degree. Then
\[
\left| xI - E_i(\alpha) \right| = (x - (1 - \alpha))^{s_i - 1} \left| xI - C_i(\alpha) \right|,
\]
where $\left| xI - C_i(\alpha) \right| = (x - (1 - \alpha))(x - (1 - \alpha)d(v_i)) - (2\alpha - 1)^2 s_i.$
\end{corollary}

\begin{proof}
 The result follows from Lemma \ref{lemS}.   
\end{proof}

\begin{remark}
    From Corollary~\ref{aux}, we have that $|\widetilde{xI - E_i(\alpha)}| = {(x - (1 - \alpha))}^{s_i},$ for $1 \leq i \leq q(G).$
\end{remark}

To simplify the notation, let $ s_q = s_{q(G)} $. Using the global labeling of the vertices of \( G \cong G(s_1, s_2, \ldots, s_{q}) \), the $B_\alpha$-matrix of $G$ can be written as

\begin{equation}\label{eq:balphamatrix}
%\scalebox{0.7}{$
B_\alpha(G) = \left[
\begin{array}{cccc}
E_1(\alpha) & \varepsilon_{1,2} \gamma R_{\scriptstyle (s_1+1)\times(s_2+1)} & \cdots & \varepsilon_{1,q} \gamma R_{\scriptstyle (s_1+1)\times(s_q+1)} \\
\varepsilon_{1,2} \gamma R_{\scriptstyle (s_1+1)\times(s_2+1)}^{\mathsf{T}} & E_2(\alpha) & \ddots & \vdots \\
\vdots & \ddots & \ddots & \varepsilon_{q-1,q} \gamma R_{\scriptstyle (s_{q-1}+1)\times(s_q+1)} \\
\varepsilon_{1,q} \gamma R_{\scriptstyle (s_1+1)\times(s_q+1)}^{\mathsf{T}} & \cdots & \varepsilon_{q-1,q} \gamma R_{\scriptstyle (s_{q-1}+1)\times(s_q+1)}^{\mathsf{T}} & E_q(\alpha)
\end{array}
\right]
%$}
\end{equation}
where $\gamma=2\alpha -1$, $R_{(s_i+1) \times (s_j+1)}$ is a matrix of order $(s_i+1) \times (s_j+1)$ whose entries
are zero except the entry in the last row and column, which equals $1$ for $1 \leq i,j \leq q(G)$ and $E_i(\alpha)$ is defined by Equation (\ref{si}).

The next theorem gives the  $B_{\alpha}$-spectra  of $G$, when $G \cong G\left( s_{1},s_{2},\ldots ,s_{q(G)}\right)$.

\begin{theorem}\label{quasi}
Let $G \cong G\left( s_{1},s_{2},\ldots ,s_{q(G)}\right)$ be a graph, such that \( s_i \geq 1 \), for \( 1 \leq i \leq q(G) \). Then, for $0 \leq \alpha \leq 1$, the $B_{\alpha}$-eigenvalues of $G$ are $1-\alpha$ with multiplicity at least $p(G)-q(G)$, and the other eigenvalues are roots of the characteristic polynomial associated with the matrix
\[
M_{i,j} =
\begin{cases}
C_i(\alpha), & \text{if } i = j, \\
\varepsilon_{i,j} \gamma R_{2 \times 2}, & \text{if } i \ne j,
\end{cases}
\]
of order $2q(G) \times 2q(G)$.
\end{theorem}

\begin{proof}
From Equation~\eqref{eq:balphamatrix}, Lemma~\ref{extra}, Corollary \ref{aux}, together with a factoring in each column, we have
\begin{equation*}
   \scalebox{0.8}{$
   |xI-B_{\alpha}(G)|=\left|\begin{array}{cccc}
    (xI -E_{1}(\alpha)) & -\varepsilon_{1,2} \gamma R_{\scriptstyle (s_1+1)\times(s_2+1)} & \ldots & -\varepsilon_{1,q(G)} \gamma R_{\scriptstyle (s_1+1)\times(s_q+1)} \\
    -\varepsilon_{1,2} \gamma R_{\scriptstyle (s_1+1)\times(s_2+1)}^{\mathsf{T}} &(xI-E_{2}(\alpha)) & \ddots & \vdots \\
    \vdots & \ddots & \ddots & -\varepsilon_{q(G)-1,q(G)}\gamma R_{\scriptstyle (s_{q-1}+1)\times(s_q+1)} \\
    -\varepsilon_{1,q(G)}\gamma R_{\scriptstyle (s_1+1)\times(s_q+1)} ^{\mathsf{T}} & \ldots & -\varepsilon_{q(G)-1,q(G)}\gamma R_{\scriptstyle (s_{q-1}+1)\times(s_q+1)}^{\mathsf{T}} & (xI-E_{q(G)}(\alpha)) \\
  \end{array}
\right|=
$}
\end{equation*}
\begin{equation*}
\scalebox{1}{$
\begin{aligned}
   &\displaystyle 
   \prod_{j=1}^{q(G)}(x-(1-\alpha))^{s_j-1}
   \left|
   \begin{array}{cccc}
       |xI-C_1(\alpha)| & -\varepsilon_{1,2}\gamma(x -(1-\alpha)) & \ldots & -\varepsilon_{1,q(G)}\gamma(x-(1-\alpha)) \\
       -\varepsilon_{1,2}\gamma(x - (1-\alpha)) & |xI-C_2(\alpha)| & \ldots & \vdots \\
       \vdots & \vdots & \ddots & -\varepsilon_{q(G)-1,q(G)}\gamma(x-(1-\alpha)) \\
       -\varepsilon_{1,q(G)}\gamma(x-(1-\alpha)) & -\varepsilon_{2,q(G)}\gamma(x-(1-\alpha)) & \ldots & |xI-C_{q(G)}(\alpha)|
   \end{array}
   \right| \\
   &= \prod_{j=1}^{q(G)}(x-(1-\alpha))^{s_j-1} \, P_M(x)
\end{aligned}
$}
\end{equation*}
Finally, we observe that $\displaystyle \prod_{j=1}^{q(G)}(x-(1-\alpha))^{s_j-1}=(x-(1-\alpha))^{p(G)-q(G)}$ and the result follows.
\end{proof}

\begin{corollary}\label{cquasi}
Let $G \cong G\left( s_{1},s_{2},\ldots ,s_{q(G)}\right)$ be a graph, such that \( s_i \geq 1 \), for \( 1 \leq i \leq q(G) \).  If, for $0 \leq \alpha \leq 1$, $\alpha \ne \dfrac{1}{2}$, then the multiplicity of $1-\alpha$ as an $B_{\alpha}$-eigenvalue of $G$ is equal to $p(G)-q(G).$
\end{corollary}
\begin{proof}
Using the same notation as in Theorem~\ref{quasi}, it is sufficient to prove that $|(1-\alpha) I-M | \neq 0.$ From the above expression for $|xI-M|$, we have \begin{equation*}
\scalebox{1}{$
|(1-\alpha) I - M| =
\left|
\begin{array}{cccc}
|(1-\alpha) I - C_1(\alpha)| & 0 & \ldots & 0 \\
0 & |(1-\alpha) I - C_2(\alpha)| & \ddots & \vdots \\
\vdots & \ddots & \ddots & 0 \\
0 & \ldots & 0 & |(1-\alpha) I - C_{q(G)}(\alpha)|
\end{array}
\right|
$}
\end{equation*}
and, for $i=1,2,\ldots,q(G)$, $|(1-\alpha) I-C_i(\alpha)|= -(2\alpha-1)^{2} s_i \neq 0$. Hence $|(1-\alpha) I-M | \neq 0.$
\end{proof}

\subsubsection{Graphs that have internal vertices which are not quasi-pendant}

In this subsection, $r(G) > q(G)$. Therefore, there are $r(G) - q(G)$ internal vertices are not quasi-pendant vertices and $q(G)$ quasi-pendant vertices are the roots of the stars $K_{1,s_{1}},\ldots,K_{1,s_{q(G)}}.$  We denote
$G \cong G(s_{1},\ldots,s_{q(G)},\textbf{0})$ as a graph, where $\textbf{0}$ indicates a vector of zeros with $r(G)-q(G)$ entries. Without loss of generality, we assume that $V_Q(G)=\{v_1,v_2,\ldots,v_{q(G)}\},$ $V_C(G)=\{v_{q(G)+1},v_{q(G)+2},\ldots,v_{r(G)}\}$ and $V_P(G)=\{v_{r(G)+1},v_{r(G)+2},\ldots,v_n\}$. 

\begin{example}To show the relationship between the labeling $v_i$ of the internal vertices and the global labeling, in Figure~\ref{fig:notquasipendant} we display a graph $G \cong G(2,3,2,0,0,0,0,0)$ where  $V_Q(G)=\{v_1,v_2,v_3\}, V_C(G)=\{v_4,v_5,v_6,v_7,v_8\}$ and $V_P(G)=\{v_9,v_{10},v_{11}, v_{12},v_{13},v_{14},v_{15}\}$. Note that, according to the general labeling we have $v_1=3, v_2=7, v_3=10,v_4=11, v_5=12, v_6=13, v_7=14, v_8=15, v_9=1, v_{10}=2, v_{11}=4, v_{12}=5, v_{13}=6, v_{14}=8$, and $v_{15}=9.$
\end{example}

%\vspace{-8mm}
\begin{figure}[H]
    \centering
    \begin{tikzpicture}[
  vertex/.style={circle,draw,fill=black,minimum size=5pt,inner sep=0pt},
  every label/.style={font=\scriptsize,label distance=0.1pt},
  scale=0.8
]

% --- linha de base ---
\node[vertex,label=above:1]  (1)  at (0,0) {};
\node[vertex,label=below:3]  (3)  at (1,0) {};
\node[vertex,label=below:12] (12) at (3,0) {};
\node[vertex,label=below:11] (11) at (4,0) {};
\node[vertex,label=below:7]  (7)  at (5,0) {};
\node[vertex,label=right:5]  (5)  at (6,0) {};

\draw (1)--(3)--(12)--(11)--(7)--(5);

% --- folha no 1 ---
\node[vertex,label=below:2] (2) at (-0.1,-0.6) {};
\draw (3)--(2);

% --- bloco esquerdo (14-15-13) ---
\node[vertex,label=left:14] (14) at (1,1.8) {};
\node[vertex,label=above:15] (15) at (1,2.5) {};
\node[vertex,label=above:13] (13) at (2,2.5) {};

\draw (3)--(14)--(15)--(13);
\draw (14)--(13);

% --- coluna do 10 e folhas 9,8 ---
\node[vertex,label=below right:10] (10) at (3,1.8) {};
\node[vertex,label=above:9]  (9)  at (2.8,2.4) {};
\node[vertex,label=right:8]  (8)  at (3.6,2.0) {};

\draw (12)--(10);
\draw (10)--(9) (10)--(8) (13)--(10);

% --- estrela no 7 (6 e 4) ---
\node[vertex,label=right:6] (6) at (5.6,1.0) {};
\node[vertex,label=right:4] (4) at (5.4,-0.8) {};

\draw (7)--(6) (7)--(4);

\end{tikzpicture}
\caption{A graph $G \cong G(2,3,2,0,0,0,0,0).$}
\label{fig:notquasipendant}
\end{figure}
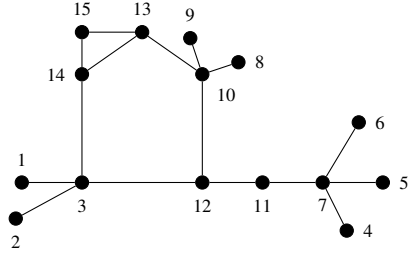 

Again, for simplicity, let \( s_q = s_{q(G)} \). We consider a convenient global labeling of the vertices of \( G \cong G(s_{1}, \ldots, s_{q}, \textbf{0}) \), starting with the pendant vertices, continuing with the quasi-pendant vertices, and ending with the internal vertices that are not quasi-pendant. Thus, the $B_{\alpha}$-matrix of $G$ can be written as follows:

\begin{equation*}
  B_{\alpha}(G)=\left[
    \begin{array}{cc}
      U & W \\
      W^{T} & N \\
    \end{array}
  \right]
\end{equation*}
where
\begin{equation*}
\scalebox{0.95}{$
U = \left[\begin{array}{ccccc}
E_1(\alpha) & \varepsilon_{1,2} \gamma R_{\scriptstyle (s_1+1)\times(s_2+1)} & \dots & \varepsilon_{1,q(G)-1} \gamma R_{\scriptstyle (s_1+1)\times s_q} & \varepsilon_{1,q(G)} \gamma R_{\scriptstyle (s_1+1)\times(s_q+1)} \\
\varepsilon_{1,2} \gamma R_{\scriptstyle (s_1+1)\times(s_2+1)}^{T} & E_2(\alpha) & \dots & \varepsilon_{2,q(G)-1} \gamma R_{\scriptstyle (s_2+1)\times s_q} & \varepsilon_{2,q(G)} \gamma R_{\scriptstyle (s_2+1)\times(s_q+1)} \\
\vdots & \vdots & \ddots & \vdots & \vdots \\
\varepsilon_{1,q(G)-1} \gamma R_{\scriptstyle (s_1+1)\times s_q}^{T} & \varepsilon_{2,q(G)-1} \gamma R_{\scriptstyle (s_2+1)\times s_q}^{T} & \dots & E_{q(G)-1}(\alpha) & \varepsilon_{q(G)-1,q(G)} \gamma R_{\scriptstyle s_q \times (s_q+1)} \\
\varepsilon_{1,q(G)} \gamma R_{\scriptstyle (s_1+1)\times(s_q+1)}^{T} & \varepsilon_{2,q(G)} \gamma R_{\scriptstyle (s_2+1)\times(s_q+1)}^{T} & \dots & \varepsilon_{q(G)-1,q(G)} \gamma R_{\scriptstyle s_q \times (s_q+1)}^{T} & E_{q(G)}(\alpha)
\end{array}\right]
$}
\end{equation*}

\begin{equation*}
W = \gamma \left[\begin{array}{ccccc}
\varepsilon_{1,q(G)+1} \mathbf{e} & \varepsilon_{1,q(G)+2} \mathbf{e} & \dots & \varepsilon_{1,r(G)-1} \mathbf{e} & \varepsilon_{1,r(G)} \mathbf{e} \\
\varepsilon_{2,q(G)+1} \mathbf{e} & \varepsilon_{2,q(G)+2} \mathbf{e} & \dots & \varepsilon_{2,r(G)-1} \mathbf{e} & \varepsilon_{2,r(G)} \mathbf{e} \\
\vdots & \vdots & \ddots & \vdots & \vdots \\
\varepsilon_{q(G)-1,q(G)+1} \mathbf{e} & \varepsilon_{q(G)-1,q(G)+2} \mathbf{e} & \dots & \varepsilon_{q(G)-1,r(G)-1} \mathbf{e} & \varepsilon_{q(G)-1,r(G)} \mathbf{e} \\
\varepsilon_{q(G),q(G)+1} \mathbf{e} & \varepsilon_{q(G),q(G)+2} \mathbf{e} & \dots & \varepsilon_{q(G),r(G)-1} \mathbf{e} & \varepsilon_{q(G),r(G)} \mathbf{e}
\end{array}\right],
\end{equation*}

\begin{equation*}
\scalebox{1}{$
N = \left[\begin{array}{ccccc}
(1 - \alpha) d(v_{q(G)+1}) & \varepsilon_{q(G)+1,q(G)+2} \gamma & \dots & \varepsilon_{q(G)+1,r(G)-1} \gamma & \varepsilon_{q(G)+1,r(G)} \gamma \\
\varepsilon_{q(G)+1,q(G)+2} \gamma & (1 - \alpha) d(v_{q(G)+2}) & \dots & \varepsilon_{q(G)+2,r(G)-1} \gamma & \varepsilon_{q(G)+2,r(G)} \gamma \\
\vdots & \vdots & \ddots & \vdots & \vdots \\
\varepsilon_{q(G)+1,r(G)-1} \gamma & \varepsilon_{q(G)+2,r(G)-1} \gamma & \dots & (1 - \alpha) d(v_{r(G)-1}) & \varepsilon_{r(G)-1,r(G)} \gamma \\
\varepsilon_{q(G)+1,r(G)} \gamma & \varepsilon_{q(G)+2,r(G)} \gamma & \dots & \varepsilon_{r(G)-1,r(G)} \gamma & (1 - \alpha) d(v_{r(G)})
\end{array}\right]
$}
\end{equation*}
where $\gamma=2\alpha -1$, $\textbf{e}$ is a column vector of zeros except for its last entry, which is $1$, $R_{(s_i+1) \times (s_j+1)}$ is a matrix of order $(s_i+1) \times (s_j+1)$ whose entries
are zero except the entry in the last row and column, which equals $1$ for $1 \leq i,j \leq q(G)$ and $d(v_{q(G)+1}), d(v_{q(G)+2}), \dots, d(v_{r(G)-1}), d(v_{r(G)})$ are the degrees of the vertices $v_{q(G)+1},v_{q(G)+2}, \ldots, v_{r(G)-1}, v_{r(G)}$, respectively.

\begin{theorem}\label{general}
Let $G \cong G\left(s_{1},s_{2},\ldots ,s_{q(G)},\mathbf{0}\right)$ be a graph, such that \( s_i \geq 1 \), for \( 1 \leq i \leq q(G) \). Then, for $0 \leq \alpha \leq 1$, the $B_{\alpha}$-eigenvalues of $G$ are $1-\alpha$ with multiplicity at least $p(G)-q(G)$
and the other eigenvalues are roots of the characteristic polynomial associated with the matrix
\[
X=\left[
\begin{array}{cc}
Q & P \\
P^{T} & N
\end{array}%
\right]
\]%
of order $(n+q(G)-p(G)) \times (n+q(G)-p(G))$ 
where
\[
%\resizebox{\textwidth}{!}{$ \scalebox{0.65}
Q=\left[\begin{array}{ccccc}
               C_1(\alpha)                   & \varepsilon_{1,2} \gamma R_{2 \times 2}     & \dots  & \varepsilon_{1,q(G)-1} \gamma R_{2 \times 2}    & \varepsilon_{1,q(G)} \gamma R_{2 \times 2} \\
               \varepsilon_{1,2}\gamma R_{2 \times 2}      & C_{2} (\alpha)                & \dots  & \varepsilon_{2,q(G)-1} \gamma R_{2 \times 2}    & \varepsilon_{2,q(G)} \gamma R_{2 \times 2}  \\
                               \vdots        & \vdots                        & \ddots & \vdots                             & \vdots  \\
               \varepsilon_{1,q(G)-1}\gamma R_{2 \times 2} & \varepsilon_{2,q(G)-1}\gamma R_{2 \times 2} & \dots  & C_{q(G)-1}(\alpha)                 & \varepsilon_{q(G)-1,q(G)} \gamma R_{2 \times 2} \\
               \varepsilon_{1,q(G)}\gamma R_{2 \times 2}   & \varepsilon_{2,q(G)}\gamma R_{2 \times 2}   & \dots  & \varepsilon_{q(G)-1,q(G)} \gamma R_{2 \times 2} & C_{q(G)}(\alpha)%
            \end{array}\right]
%$}
\]
 \[
%\resizebox{\textwidth}{!}{$
N=\left[\begin{array}{ccccc}
        (1-\alpha) d(v_{_{q(G)+1}})    & \varepsilon_{_{q(G)+1,q(G)+2}}\gamma & \dots  & \varepsilon_{_{q(G)+1,r(G)-1}}\gamma & \varepsilon_{_{q(G)+1,r(G)}}\gamma \\
        \varepsilon_{_{q(G)+1,q(G)+2}}\gamma & (1-\alpha) d(v_{_{q(G)+2}})     & \dots  & \varepsilon_{_{q(G)+2,r(G)-1}}\gamma & \varepsilon_{_{q(G)+2,r(G)}}\gamma \\
        \vdots                              & \vdots                              & \ddots & \vdots                              & \vdots\\
        \varepsilon_{_{q(G)+1,r(G)-1}}\gamma & \varepsilon_{_{q(G)+2,r(G)-1}}\gamma & \dots  & (1-\alpha) d(v_{_{r(G)-1}})   & \varepsilon_{_{r(G)-1,r(G)}}\gamma \\
        \varepsilon_{_{q(G)+1,r(G)}}\gamma   & \varepsilon_{_{q(G)+2,r(G)}}\gamma & \dots  & \varepsilon _{_{r(G)-1,r(G)}}\gamma & (1-\alpha) d(v_{_{r(G)}})%
\end{array}%
\right]
%$}
\]
and
\[
P=\gamma \left[\begin{array}{ccccc}
              \varepsilon_{1,q(G)+1}\mathbf{e}      & \varepsilon_{1,q(G)+2}\mathbf{e}     & \dots  & \varepsilon_{1,r(G)-1}\mathbf{e}     & \varepsilon_{1,r(G)}\mathbf{e} \\
              \varepsilon_{2,q(G)+1}\mathbf{e}      & \varepsilon_{2,q(G)+2}\mathbf{e}     & \dots  & \varepsilon_{2,r(G)-1}\mathbf{e}     & \varepsilon_{2,r(G)}\mathbf{e} \\
              \vdots                                & \vdots                               & \ddots & \vdots                               & \vdots \\
              \varepsilon_{q(G)-1,q(G)+1}\mathbf{e} & \varepsilon_{q(G)-1,q(G)+2}\mathbf{e}& \dots  & \varepsilon_{q(G)-1,r(G)-1}\mathbf{e}& \varepsilon_{q(G)-1,r(G)}\mathbf{e} \\
              \varepsilon_{q(G),q(G)+1}\mathbf{e}   & \varepsilon_{q(G),q(G)+2}\mathbf{e}  & \dots  & \varepsilon_{q(G),r(G)-1}\mathbf{e}  & \varepsilon_{q\left(G\right),r(G)}\mathbf{e}%
\end{array}\right].
\]%
\end{theorem}
\begin{proof}
    From Lemma~\ref{extra}, Corollary~\ref {extra}, and Corollary~\ref {aux}, together with factoring in each of the first $q(G)$ columns of
the resulting determinant, the result follows.
\end{proof}

The next theorem displays the expression for the multiplicity of $1-\alpha$ as $B_{\alpha}$-eigenvalue of $G$ when $G \cong G\left( s_{1},s_{2},\ldots ,s_{q(G)}, \mathbf{0}\right)$.

\begin{theorem}\label{nscond}
Let $G \cong G\left( s_{1},s_{2},\ldots ,s_{q(G)},\mathbf{0}\right)$ be a graph, such that \( s_i \geq 1 \), for \( 1 \leq i \leq q(G) \). Let $X, Q$ and $N$ be the matrices in Theorem \ref{general}. If $0 \leq \alpha \leq 1$ with $\alpha \ne \dfrac{1}{2}$, then
\begin{enumerate}[(i)]
\item $m_{X}(1-\alpha)=m_{N}(1-\alpha)$.
\item $m_{B_{\alpha}(G)}(1-\alpha)=p(G)-q(G)+m_{N}(1-\alpha).$
\end{enumerate}
\end{theorem}
\begin{proof}
Suppose $G \cong G\left( s_{1},s_{2},\ldots ,s_{q(G)},\mathbf{0}\right)$ and $0 \leq \alpha \leq 1$ with $\alpha \ne \dfrac{1}{2}$. So,

$$\left\vert (1-\alpha) I-X \right\vert =\left\vert \begin{array}{cc}
      (1-\alpha) I-Q & -P \\
      -P^T & (1-\alpha) I-N
\end{array}\right\vert.$$
From Corollary \ref{extra}, we obtain

\begin{equation*}
\scalebox{1}{$
\left| (1 - \alpha) I - X \right| =
\left|
\begin{array}{ccccc}
-(2\alpha - 1)^2 s_1 & 0 & \cdots & 0 & * \\
0 & -(2\alpha - 1)^2 s_2 & \cdots & 0 & * \\
\vdots & \vdots & \ddots & \vdots & \vdots \\
0 & 0 & \cdots & -(2\alpha - 1)^2 s_{q(G)} & * \\
0 & 0 & \cdots & 0 & \left| (1 - \alpha) I - N \right|
\end{array}
\right|
$}
\end{equation*}
\noindent
$\textstyle = (-1)^{q(G)} (1 - \alpha)^{2q(G)} s_1 s_2 \cdots s_{q(G)} \left| \alpha I - N \right|$, where $*$ denotes some value obtained after performing elementary operations. 

As $2\alpha - 1 \neq 0$ and $s_{i}\geq 1$, for $1\leq i\leq q(G)$, we obtain $\left\vert (1-\alpha) I-X \right\vert = 0$ if and only if $\left\vert (1-\alpha) I-N \right\vert = 0$. Hence $m_{X}(1-\alpha)=m_{N}(1-\alpha)$ Moreover, from Theorem \ref{general}, $m_{B_{\alpha}(G)}(1-\alpha)=p(G)-q(G)+m_{X}(1-\alpha)$. As $m_{X}(1-\alpha)=m_{N}(1-\alpha)$, the result follows.
\end{proof}

We can see that the matrix $N$ in Theorem \ref{general} is obtained from $B_{\alpha}(G)$ by taking the entries corresponding to the internal vertices which are not quasi-pendant vertices.

\begin{theorem}\label{alfa}
Let $G \cong G(s_{1}, s_{2}, \ldots, s_{q(G)}, \mathbf{0})$ be a graph, such that \( s_i \geq 1 \), for \( 1 \leq i \leq q(G) \). Let $X$, $Q$, and $N$ be the matrices defined in Theorem~\ref{general}, and let $G_1, \ldots, G_t$ be the connected components of the subgraph of $G$ induced by $V_C(G)$. If $0 \leq \alpha \leq 1$ with $\alpha \ne \dfrac{1}{2}$,
\begin{center}
$m_{B_{\alpha}(G)}(1-\alpha) = p(G) - q(G) + \displaystyle\sum_{i=1}^{t} m_{N_i}(1 - \alpha)$,
\end{center}
for $1 \le i \le t$, where
\[
N_i = (2\alpha - 1)  A(G_i) + (1 - \alpha) D_i,
\]
with $D_i$ being the diagonal matrix of order $|V(G_i)|$ whose diagonal entries are the degrees in $G$ of the corresponding vertices.
\end{theorem}

\begin{proof}
There is a labeling of the vertices of $V_C(G)$ such that $\displaystyle N=\bigoplus_{i=1}^{t}N_i$. 
Therefore $\displaystyle m_{N}(1-\alpha)=\sum_{i=1}^{t} m_{N_{i}}(1-\alpha)$ and, by Theorem \ref{nscond}, the result follows.
\end{proof}

As a consequence of Theorem \ref{alfa}, we have the following corollary.

\begin{corollary}
Let $G \cong G(s_{1}, s_{2}, \ldots, s_{q(G)}, \mathbf{0})$ be a graph such that \( s_i \geq 1 \), for \( 1 \leq i \leq q(G) \), and $G_1,\ldots, G_t$ be the connected components of the subgraph induced of $G$ by $V_C(G)$. Then:
\begin{enumerate}[(i)]
\item $\displaystyle \eta (G) = p(G)-q(G)+\sum_{i=1}^{t}\eta(G_i)$,
\item $\displaystyle  m_{L(G)}(1) = p(G)-q(G)+\sum_{i=1}^{t}m_{N_i}(1)$, with $N_i = D_i-A(G_i)$, and
\item $\displaystyle  m_{Q(G)}(1) = p(G)-q(G)+\sum_{i=1}^{t}m_{N_i}\big(\frac{1}{3}\big)$, with
$N_i = \frac{1}{3} (D_i+A(G_i))$, 
\end{enumerate}
where $D_i$ being the diagonal matrix of order $|V(G_i)|$ whose diagonal entries are the degrees in $G$ of the corresponding vertices. 
\end{corollary}

\section{Positive semidefiniteness of \texorpdfstring{$B_{\alpha}$}{Balpha}-matrices of graphs}\label{semidef}

The Laplacian and signless Laplacian matrices are positive semidefinite but this is not true for $B_{\alpha}(G)$ when $\alpha$ is sufficiently large. In 2024, Samanta $\it{et}$ $\it{al.}$ \cite{Samanta2024} defined $\beta_0(G):=\max\{ \alpha \in (0,1) :  \lambda_{n}(B_\alpha(G))=0 \}$ and proved that if $G$ is a graph with no isolated vertices, then $\beta_0(G) \geq \frac{2}{3}$ and, additionally, $B_{\alpha}(G)$ is positive semidefinite if and only if $\alpha \in [0,\beta_0(G)]$. Thus, we define the following problem.

\begin{problem}\label{prob2}
Given a graph $G$, find ${\beta}_{0}(G)$.
\end{problem}
This problem was motivated by the study of the positive semidefiniteness of $A_\alpha(G)$. Nikiforov and Rojo~\cite{NiRo2017} defined ${\alpha}_{0}(G)$ as the smallest value in the interval $[0,1]$ such that ${\lambda}_{n}(A_{{\alpha}_0(G)}(G)) = 0$. Then, $A_{\alpha}(G)$ is positive semidefinite if and only if ${\alpha}_{0}(G) \leq \alpha \leq 1$. Thus, they raised the following problem: {\it{Given a graph $G$, find ${\alpha}_{0}(G)$}}. A major difference in addressing these problems is the nonmonotonicity of $B_{\alpha}$-eigenvalues of $G$ in $\alpha \in [0,1]$. 

In this context, the following result solves Problem \ref{prob2} when $G$ is bipartite.

\begin{proposition}
A graph $G$ is bipartite, if and only if $\beta_0(G)=\frac{2}{3}$.
\end{proposition}

\begin{proof}
Let $G$ be a bipartite graph of $n$ vertices with $m$ edges. For $\alpha \in [0,1]$, from Theorem 4.13 in \cite{Samanta2024}, we have that
\begin{equation*}
    \lambda_n(B_{\alpha}(G)) \leq \frac{2m}{n}(2-3\alpha),
\end{equation*}
where the equality occurs if and only if either $\beta_0(G)=\frac{2}{3}$, or $G$ is regular with $\alpha \geq \frac{1}{2}$. Now, consider $\alpha=\beta_0(G)$. By definition of $\beta_0(G)$, we obtain
\begin{equation*}
    0=\lambda_n(B_{\beta_0(G)}(G)) \leq \frac{2m}{n}(2-3\beta_0(G)),
\end{equation*}
where the equality occurs if and only if either $\beta_0(G)=\frac{2}{3}$ or $G$ is a graph without edges. Since $G$ is bipartite, $m>1$, then the equality occurs if and only if $\beta_0(G)=\frac{2}{3}$.

Now, suppose that $\beta_0(G)=\frac{2}{3}$. As $B_{\frac{2}{3}}(G)=\frac{1}{3}Q(G)$, we have
\begin{equation*}
 0=\lambda_n(B_{\frac{2}{3}}(G))=\frac{1}{3}\lambda_n(Q(G)),  
\end{equation*}
which occurs if and only if $G$ is bipartite, and the result follows.
\end{proof}

In \cite{Samanta2024}, Problem \ref{prob2} has been solved for families of regular graphs and regular non-bipartite graphs, and interesting relationships between $\beta_0(G)$ and combinatorial invariants, such as the chromatic number and independence number of $G$, are obtained. However, the problem remains open for families of non-bipartite graphs. In this context, we solve the problem for the following family of non-bipartite graphs. As usual, $K_{n}$ denotes the complete graph of order $n$. Given integers \( n \) and \( \ell \), where \( n \geq 3 \) and \( 1 \leq \ell \leq n \), let \( \mathcal{H}^\ell_n \) be the graph constructed from two copies of the complete graph \( K_n \) by adding \( \ell \) edges between \( \ell \) distinct vertices of a copy of \( K_n \) and their corresponding vertices in the other copy. To illustrate, Figure~\ref{HK} shows the graph $\mathcal{H}_6^3$. Propositions \ref{theo::H_ln} and \ref{prop::Hln_beta}, exhibit the $B_{\alpha}$-spectra of $G$ and $\beta_0(G)$ when $G \cong \mathcal{H}^\ell_n$, respectively.

\begin{figure}[H]
    \centering
    %\begin{subfigure}[b]{0.3\textwidth}
       % \centering
       % \includegraphics[width=0.7\linewidth]{pineapple.png}
        %\caption{Graph $K_5^3$}
       % \label{fig:pineapple}
   % \end{subfigure}
    %\hspace{0.001\textwidth}
    \begin{subfigure}[b]{0.4\textwidth}
        \centering
        \includegraphics[width=\linewidth]{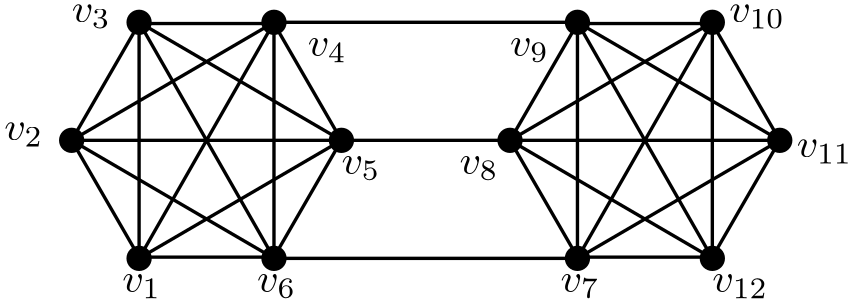}
        %\caption{Graph $H_6^3$}
        %\label{fig:grafoH}
    \end{subfigure}
    \hspace{0.06\textwidth}
   % \begin{subfigure}[b]{0.3\textwidth}
        %\centering
    %\includegraphics[width=1.01\linewidth]{grafoKK.png}
       % \caption{Graph $KK_5^3$}
       % \label{fig:grafoKK}
    %\end{subfigure}
    \caption{Graph $\mathcal{H}_6^3$.}
    \label{HK}
\end{figure}

\begin{proposition}\label{theo::H_ln}
 Let $G \cong \mathcal{H}_{n}^{\ell}$ of order $2n$ and $\alpha \in [0,1].$ If $1 \leq \ell < n,$  then
\[
%\resizebox{\textwidth}{!}{$
\sigma(B_{\alpha}(G)) = \{{(n - \alpha(n+1))}^{2n-2\ell-2}, {(n+2- \alpha(n+4))}^{\ell-1}, {(n-\alpha n)}^ {\ell-1}, \theta_1,\theta_2,\theta_3, \theta_4 
%\}$}
\] where
\[
%\resizebox{\textwidth}{!}{$
\begin{aligned}
    \theta_1 &= \dfrac{1}{2}\left( n - \alpha - \sqrt{(2\alpha - 1)(4\alpha \ell - 2\alpha n) + n^{2}(2\alpha - 1)^{2} + \alpha^{2}} \right)\\
    \theta_2 &= \dfrac{1}{2}\left( n - \alpha + \sqrt{(2\alpha - 1)(4\alpha \ell - 2\alpha n) + n^{2}(2\alpha - 1)^{2} + \alpha^{2}} \right)\\
    \theta_3 &= \dfrac{1}{2} \left( n -5\alpha + 2 -\sqrt{(2\alpha - 1)\left[(3\alpha - 2)(-4\ell + 2n) + n^{2}(2\alpha - 1)\right] + (3\alpha - 2)^{2}} \right) \\
    \theta_4 &= \dfrac{1}{2} \left(n -5\alpha +2 + \sqrt{(2\alpha - 1)\left[(3\alpha - 2)(-4\ell + 2n) + n^{2}(2\alpha - 1)\right] + (3\alpha - 2)^{2}} \right)
\end{aligned}
%$}
\]
If $\ell=n$, then
{\footnotesize{$\sigma(B_{\alpha}(G)) = \{n\alpha, 2+ \alpha(n-4), {(n+2-\alpha (n+4))}^ {n-1}, {(n - \alpha n)}^{n-1} \}.$}}
\end{proposition}
\begin{proof}
   Consider $G \cong \mathcal{H}_{n}^{\ell}$ with $1 \leq \ell < n$. With a convenient labeling of the vertices of $G$, we can write the $B_{\alpha}$-matrix of $G$ the following way
   \[
    %\resizebox{\textwidth}{!}{$
    B_\alpha\left(G\right) = \left[
	\begin{array}{c:c:c:c}
		B_{11} & (2\alpha - 1)J_{ \left(n-\ell\right) \times \ell} & {\bf{0}}_{\left(n-\ell\right) \times \ell} & {\bf{0}}_{\left(n-\ell\right) \times \left(n-\ell\right)}\\
		\hdashline
		(2\alpha - 1)J^T_{ \left(n-\ell\right) \times \ell} & B_{22} & B_{23} & {\bf{0}}_{\ell \times \left(n-\ell\right)}\\ 
        \hdashline
		{\bf{0}}^T_{(n-\ell) \times \ell} & B_{32} & B_{33}  & (2\alpha - 1)J_{\ell \times \left(n-\ell\right)}\\ 
        \hdashline
		{\bf{0}}^T_{\left(n-\ell\right) \times \left(n-\ell\right)} & {\bf{0}}^T_{\ell \times \left(n-\ell\right)} & (2\alpha - 1)J^T_{\ell \times \left(n-\ell\right)} & B_{44}
	\end{array} 
	\right],
    %$}
    \]
where 
\begin{eqnarray*}
B_{11} &=& B_{44} = (1-\alpha)(n-1)I_{n-\ell}+(2\alpha - 1)(J-I)_{n-\ell}, \\
B_{22} &=& B_{33} = n(1-\alpha)I_{\ell\times \ell} + (2\alpha-1)(J-I)_{\ell \times \ell}, \text{\,\ and} \\ 
B_{23} &=& B_{32} =(2\alpha -1)I_{\ell \times \ell}.
\end{eqnarray*}
Considering this labeling, we obtain an equitable quotient matrix given by

\[
\resizebox{\textwidth}{!}{$
{\overline{B_{\alpha} (G)}} = \begin{bmatrix}
\ell + (n-1-2\ell)\alpha & \ell(2\alpha - 1) & 0 & 0 \\
(n - \ell)(2\alpha - 1) & n - \ell + 1 + (2\ell - 2 - n)\alpha & 2\alpha - 1 & 0 \\
0 & 2\alpha - 1 & n - \ell + 1 + (2\ell - 2 - n)\alpha & (n - \ell)(2\alpha - 1) \\
0 & 0 & \ell(2\alpha - 1) & \ell + (n - 1 - 2\ell)\alpha
\end{bmatrix}
$}.
\]
Computing the eigenvalues of ${\overline{B_{\alpha} (G)}}$, we obtain:
\[
%\resizebox{\textwidth}{!}{$
\begin{aligned}
    \theta_1 &= \dfrac{1}{2}\left( n - \alpha - \sqrt{(2\alpha - 1)(4\alpha \ell - 2\alpha n) + n^{2}(2\alpha - 1)^{2} + \alpha^{2}} \right)\\
    \theta_2 &= \dfrac{1}{2}\left( n - \alpha + \sqrt{(2\alpha - 1)(4\alpha \ell - 2\alpha n) + n^{2}(2\alpha - 1)^{2} + \alpha^{2}} \right)\\
    \theta_3 &= \dfrac{1}{2} \left( n -5\alpha +2 -\sqrt{(2\alpha - 1)\left[(3\alpha - 2)(-4\ell + 2n) + n^{2}(2\alpha - 1)\right] + (3\alpha - 2)^{2}} \right) \\
    \theta_4 &= \dfrac{1}{2} \left(n -5\alpha + 2 + \sqrt{(2\alpha - 1)\left[(3\alpha - 2)(-4\ell + 2n) + n^{2}(2\alpha - 1)\right] + (3\alpha - 2)^{2}} \right).
\end{aligned}
%$}
\]
From Lemma \ref{quotient}, $\sigma(\overline{B_{\alpha} (G)}) \subseteq \sigma(B_{\alpha}(G))$. If $1 \leq \ell < n$, by Corollary \ref{twinsvertices}, we conclude that $n - \alpha(n+1)$ is a $B_\alpha$-eigenvalue of $G$ with multiplicity at least $2n - 2\ell -2$. 
Hereafter, let $e_k$ be the $k$-th vector of the canonical basis of $\real^{2n}$.
Consider vectors $e_j - e_\ell + e_{n+j} - e_{n+\ell}, \ \ j = 1, \ldots, \ell-1.$ It is easy to see that they are the eigenvectors of $B_\alpha(G)$ associated with the eigenvalue $ n-n\alpha.$ So, this eigenvalue has the multiplicity at least $\ell-1$. 
Finally, the vectors $e_j-e_\ell -e_{n+j} + e_{n+\ell} , \ \ j = 1, \ldots, \ell-1$,  are eigenvectors of $B_\alpha(G)$ associated with the eigenvalue $(n+2) - (n+4)\alpha$ with multiplicity at least $\ell-1$. Thus, the $B_{\alpha}$-spectra of $G$ is obtained, when $1 \leq \ell < n$. \newline 
On the other hand, consider $G \cong \mathcal{H}_{n}^{n}$ ($\ell=n$). It is easy to see that the all ones vector
(vector whose all entries are equal to $1$) is an eigenvector of $B_{\alpha}(G)$ associated with the eigenvalue $n\alpha$. Moreover, the vector $\displaystyle \sum_{j=1}^n (e_{n+j}-e_j)$, is an eigenvector of $B_{\alpha}(G)$ associated to the eigenvalue $2+ \alpha(n-4)$. The remaining eigenvalues are obtained similarly to the case when $1 \leq \ell < n,$ and the result follows.
\end{proof}

\begin{proposition}\label{prop::Hln_beta}
    Let $\alpha \in [0,1)$ and $G \cong \mathcal{H}_{n}^{\ell}$ for $n \geq 3$. Then 
    \begin{enumerate}[(i)]
        \item \text{if} $1 < \ell \leq n,$ \ $\beta_0(G) = \dfrac{n+2}{n+4};$
        
        \item if $\ell = 1,$ $\beta_0(G) = \dfrac{n^2+n-9+\sqrt{n^4+2n^3-9n^2+6n+1}}{\,2(n^2+3n-10)\,}$ .  
    \end{enumerate} 
\end{proposition}

\begin{proof}
For determine the $\beta_0(\mathcal{H}_{n}^{\ell}),$ the proof is divided into three cases. In the first case, $\ell = n$, the result follows immediately. Now, suppose $1 < \ell < n$. From Proposition \ref{theo::H_ln}, we can observe that the eigenvalues $n-\alpha n$, $\theta_2$, and $\theta_4$ are positive for $\alpha \in [0,1)$.
%Assume throughout that $n \geq 3$ and $1 \leq \ell < n$. 
In addition, the eigenvalue $n-\alpha(n+1)$ is nonnegative for $\displaystyle \alpha \in \left[0, \dfrac{n}{n+1}\right]$, 
and the eigenvalue $n+2 - \alpha(n+4)$ is nonnegative for $\alpha \in \left[0, \dfrac{n+2}{n+4}\right]$. From some algebraic manipulations, we can see that the eigenvalue $\theta_1$ is nonnegative for $\alpha \in \left[0, \dfrac{n^2-n+\ell}{n^2-n+2\ell}\right]$, and the eigenvalue $\theta_3$ is nonnegative for 
$\alpha \in \left[0, f(n,\ell)\right]$, where
\[
f(n,\ell) = \frac12 - \frac{\,2n+\ell-2-\sqrt{(2n+\ell-2)^2+n(n-1)(n^2+3n-6\ell-4)}\,}{\,2(n^2+3n-6\ell-4)}.
\]
and it is verified that $\dfrac{n}{n+1} - \dfrac{n+2}{n+4} \geq 0$, $\dfrac{n^2-n+\ell}{n^2-n+2\ell} - \dfrac{n+2}{n+4} \geq 0$, and $f(n,\ell) - \dfrac{n+2}{n+4} \geq 0$. Consequently, we conclude that $\beta_0(G) = \dfrac{n+2}{n+4}$. Finally, suppose that $\ell = 1$. From Proposition \ref{theo::H_ln}, we have
\[
\sigma\!\left(B_{\alpha}(G)\right)
= \Bigl\{\bigl(n-\alpha(n+1)\bigr)^{2n-4},\ \theta_1,\ \theta_2,\ \theta_3,\ \theta_4 \Bigr\}, \text{ where }
\]
\[
\scalebox{1}{$
\begin{aligned}
\theta_1 &= \tfrac{1}{2}\!\left(
n-\alpha - 
\sqrt{(2\alpha-1)(4\alpha-2\alpha n)+n^2(2\alpha-1)^2+\alpha^2}
\right),\\[4pt]
\theta_2 &= \tfrac{1}{2}\!\left(
n-\alpha + 
\sqrt{(2\alpha-1)(4\alpha-2\alpha n)+n^2(2\alpha-1)^2+\alpha^2}
\right),\\[4pt]
\theta_3 &= \tfrac{1}{2}\!\left(
n-5\alpha+2 - 
\sqrt{(2\alpha-1)\!\big[(3\alpha-2)(-4+2n)+n^2(2\alpha-1)\big]+(3\alpha-2)^2}
\right),\\[4pt]
\theta_4 &= \tfrac{1}{2}\!\left(
n-5\alpha+2 + 
\sqrt{(2\alpha-1)\!\big[(3\alpha-2)(-4+2n)+n^2(2\alpha-1)\big]+(3\alpha-2)^2}
\right).
\end{aligned}
$}
\]
Note that $\theta_2$ and $\theta_4$ are nonnegative for $\alpha \in [0,1)$. Then, analyzing similarly to the previous case, we obtain
\[
\beta_0(G) = 
\dfrac{n^2+n-9+\sqrt{n^4+2n^3-9n^2+6n+1}}
      {\,2(n^2+3n-10)}
\]
and the result follows.
\end{proof}

The following remark displays the $\alpha_0(G)$ and $\beta_0(G)$ of some families known in the literature.

\begin{remark}\label{remarkab} \begin{enumerate}[(i)]
    \item From Samanta \textit{et} \textit{al.} \cite{Samanta2024}, $ \lambda_n(B_{\alpha}(K_n)) = n - \alpha n - \alpha$. Then, $\beta_0(K_n)=\dfrac{n}{n+1}$.
    \item  From Nikiforov \cite{Niki2017},  $\alpha_0(K_n)=\dfrac{1}{n}$.
     \item From Brondani \textit{et} \textit{al.} \cite{AFC2022}, if $\ell=1$, \begin{equation*} \alpha_0(\mathcal{H}_{n}^{1})= \dfrac{-n^{2}-n+5+\sqrt{n(n-1)(n^2 +3n -6)+1}}{4},
     \end{equation*} and if $2 \leq \ell \leq n,$ $\alpha_0(\mathcal{H}_{n}^{\ell})=\dfrac{2}{n+2}.$
     
     \end{enumerate}  
\end{remark}

Since $\alpha_0(G) \in (0,\frac{1}{2}]$ and $\beta_0(G) \in [\frac{2}{3},1)$, there exists a parameter $\epsilon(G)$ that depends on $G$ such that $\beta_0(G) = \alpha_0(G) + \epsilon(G).$ Thus, we present the following problem.

\begin{problem}
 Given a graph $G$, find the parameter $\epsilon(G)$ such that $\beta_0(G)=\alpha_0(G)+\epsilon(G).$  
\end{problem}

Note that the largest value of $\alpha_0(G)$ and the smallest value of $\beta_0(G)$ are $\frac{1}{2}$ and $\frac{2}{3}$, respectively. These values are achieved when $G$ is bipartite, then $\epsilon(G) \geq \frac{1}{6}$ where the equality occurs if and only if $G$ is bipartite.  In the following proposition, we obtain $\epsilon(G)$ for some families of graphs depending only on the order of $G$, and whose proof follows from the definition of parameter $\epsilon(G)$ and Remark \ref{remarkab}. 

\begin{proposition} 
Let $G \cong K_n$ or $G \cong \mathcal{H}_{n}^{\ell}$, with $n\geq3$. Then,   
\begin{enumerate}[(i)]         
\item $\epsilon(K_n)=\dfrac{n^2-n-1}{n(n+1)};$

\item if $\ell=1,$ $\epsilon(\mathcal{H}_{n}^{1})= \dfrac{n^2+n-3}{4}+\dfrac{2-4n+(12-n^2-3n)\sqrt{n^4+2n^3-9n^2+6n+1}}{4(n+5)(n-2)}$, and 
\item if $2 \leq \ell \leq n,$ $\epsilon(\mathcal{H}_{n}^{\ell})=\dfrac{n}{n+4} - \dfrac{4}{(n+4)(n+2)}$.        
\end{enumerate}
\end{proposition}

\section{Conclusions}

An important difference between $A_{\alpha}$-matrix and $B_{\alpha}$-matrix lies in the fact that there is no monotonicity of the $B_{\alpha}$-eigenvalues for $\alpha \in [0,1]$. However, we proved the convexity of $\lambda_1(B_{\alpha}(G))$ and the concavity of $\lambda_n(B_{\alpha}(G))$ for $\alpha \in [0,1]$, and some inequalities of the $B_{\alpha}$-eigenvalues, with emphasis on $\lambda_1(B_{\alpha}(G))$, when operations are performed on the edges of the graph for some subintervals of $[0,1]$. Additionally, an identity is provided that relates the $A_{\alpha}$-eigenvalues and $B_{\alpha}$-eigenvalues.

We obtain lower bounds for the multiplicities of some $B_{\alpha}$-eigenvalues of $G$ when the sets of vertices are independent sets, cliques, false twins, or true twins. From this context, Problem \ref{prob1pendant} was posed and solved for sets of independent vertices that consist only of pendant vertices. Thus, this work provides an exact expression of the multiplicity of $1-\alpha$ as $B_{\alpha}$-eigenvalue of graphs with pendant vertices, which will allow us to unify the study of the multiplicity of $1$ as a Laplacian and signless Laplacian eigenvalue, and the nullity of a graph.

Finally, the problem of positive semidefiniteness of $B_{\alpha}(G)$ is defined and solved for bipartite graphs. Since this problem remains open for non-bipartite graphs, we studied the family of graphs $\mathcal{H}_n^{\ell}$, obtaining their $B_{\alpha}$-spectra and  $\beta_0(\mathcal{H}_n^{\ell})$. Moreover, we present a problem that allows us to establish a relationship between $\alpha_0(G)$ and $\beta_0(G)$ through a parameter $\epsilon(G)$, and we provide this parameter for some families of graphs. 

\section*{Acknowledgments}

\vspace{0.2cm}

\noindent \textbf{Funding information}: This study was partially funded by FAPERJ - Fundação Carlos Chagas Filho de Amparo à Pesquisa do Estado do Rio de Janeiro, Process SEI 260003/001228/2023, and by CNPq-Conselho Nacional de Desenvolvimento Científico e Tecnológico, Grant 405552/2023-8. The research of G. Past\'en was partially supported by DGI-UAntof (CR 1820); Facultad de Ciencias Básicas of the Universidad de Antofagasta (CR 2201), and Departamento de Matemáticas of the Universidad de Antofagasta (CR 2203).

%\bibliographystyle{elsarticle-harv}
%\bibliography{bibliog}

\end{document}